\documentclass[12pt]{article}
\usepackage[margin=2.9cm]{geometry}
\usepackage{algorithmicx}

\usepackage[T1]{fontenc}
\usepackage[utf8]{inputenc}
\usepackage[english]{babel}
\usepackage{amsmath,amssymb, amsthm}
\usepackage[bookmarksopen,colorlinks,linkcolor=blue]{hyperref}
\usepackage[ruled,vlined]{algorithm2e}
\usepackage{authblk}
\usepackage[shortlabels]{enumitem}
\usepackage{graphicx}
\usepackage{enumitem}

\usepackage{lmodern}
\usepackage{color}

\newcommand{\OptPol}{\mathbf{OptPol}}
\newcommand{\PolPoten}{\mathbf{PolPoten}}

\title{Polyhedral value iteration for discounted games and energy games} 

\newtheorem{theorem}{Theorem}

\newtheorem{lemma}[theorem]{Lemma}
\newtheorem{proposition}[theorem]{Proposition}

\newtheorem{assumption}{Assumption}

\newtheorem{remark}{Remark}

\renewcommand{\le}{\leqslant}
\renewcommand{\ge}{\geqslant}

\newcommand{\Min}{\mathrm{Min}}
\newcommand{\Max}{\mathrm{Max}}

\def\Max{\mathrm{Max}}
\def\Min{\mathrm{Min}}

\author{Alexander Kozachinskiy\thanks{Alexander.Kozachinskiy@warwick.ac.uk.  Supported by the EPSRC grant EP/P020992/1 (Solving Parity Games in Theory and Practice).}
}

\affil{Department of Computer Science, University of Warwick, Coventry, UK}

\begin{document}

\maketitle

\begin{abstract}
We present a deterministic algorithm, solving discounted games with $n$ nodes in  $n^{O(1)}\cdot (2 + \sqrt{2})^n$-time. For bipartite discounted games our algorithm runs in $n^{O(1)}\cdot 2^n$-time.
Prior to our work no deterministic algorithm running in time $2^{o(n\log n)}$ regardless of the discount factor was known.

We call our approach polyhedral value iteration. We rely on a well-known fact that the values of a discounted game can be found from the so-called optimality equations. In the algorithm we consider a polyhedron obtained by relaxing optimality equations. We iterate points on the border of this polyhedron by moving each time along a carefully chosen shift as far as possible. This continues until the current point satisfies optimality equations.

Our approach is heavily inspired by a recent algorithm of Dorfman et al.~(ICALP 2019) for energy games. For completeness, we present their algorithm  in terms of polyhedral value iteration. Our exposition, unlike the original algorithm, does not require edge weights to be integers and works for arbitrary real weights.

\end{abstract}

\section{Introduction}

We study \textbf{discounted games}, \textbf{mean payoff games} and \textbf{energy games}. All these three kinds of games are played on finite weighted directed graphs between two players called $\Max$ and $\Min$. In case of the discounted games, description of a game graph also includes a real number $\lambda\in(0, 1)$ called the \emph{discount factor}.
 Players shift a pebble along the edges of a graph. Nodes of the graph are partitioned into two subsets, one 
where $\Max$ controls the pebble and the other where $\Min$ controls the pebble. One should also indicate in advance a starting node (a node where the pebble is located initially).
By making infinitely many moves the players give rise to an infinite sequence of edges $e_1, e_2, e_3, \ldots$ of the graph (here $e_i$ is the $i$th edge passed by the pebble).
The outcome of the game is a real number determined by a sequence $w_1, w_2, w_3, \ldots$, where $w_i$ is the weight of the edge $e_i$. We assume that outcome serves as the amount of fine paid by player $\Min$ to player $\Max$. In other words, the goal of $\Max$ is to maximize the outcome and the goal of $\Min$ is to minimize it.

The outcome is computed differently in discounted, mean payoff and energy games.
\begin{itemize}
\item the outcome of a \textbf{discounted game} is
$$\sum\limits_{i = 1}^\infty \lambda^{i - 1} w_i,$$
where $\lambda$ is the discount factor of our game graph.
\item the outcome of a \textbf{mean payoff game} is
$$\limsup\limits_{n\to\infty} \frac{w_1 + \ldots + w_n}{n}.$$
\item the outcome of an \textbf{energy game} is
$$\begin{cases}1 & \mbox{the sequence $(w_1 + w_2 + \ldots + w_n), n\in\mathbb{N}$ is bounded from below,} \\ 0 & \mbox{otherwise,}\end{cases}$$
(we  interpret outcome $1$ as victory of $\Max$ and outcome $0$ as victory of $\Min$).
\end{itemize}
All these three games are \emph{determined}. This means that in every game graph, for every node $v$ of the graph there exists a real number $\alpha$ (called the \emph{value} of $v$) such that
\begin{itemize} \item \textbf{\emph{(a)}} there is a $\Max$'s strategy $\sigma$, guarantying that the outcome is at least $\alpha$ if the game starts in $v$; 
\item \textbf{\emph{(b)}} there is a $\Min$'s strategy $\tau$, guarantying that the outcome is at most $\alpha$ if the game starts in $v$.
\end{itemize}
Any such pair of strategies $(\sigma, \tau)$ is called \emph{optimal} for $v$. 

 Moreover~\cite{shapley1953stochastic,ehrenfeucht1979positional,chakrabarti2003resource}, these games are \emph{positionally determined}. This means that for every game graph there is a pair of \emph{positional} strategies $(\sigma, \tau)$ which is optimal for every node of the graph. A strategy is positional if for every node $v$ it always makes the same move when the pebble is in $v$.


We study algorithmic problems that arise from these games. Namely, the \emph{value problem} is a problem of finding values of the nodes of a given game graph. The \emph{decision problem} is a problem of comparing the value of a node with a given threshold. Another fundamental problem, called \emph{strategy synthesis}, is to find optimal positional strategies.

\bigskip

\textbf{Motivation}. Positionally determined games are of great interest in the  design of algorithms and computational complexity. Specifically, these games serve as a source of problems that are in NP$\cap$coNP but not known to be in P.

Below we survey algorithms for discounted, mean payoff and energy games (including our contribution). Mean payoff and discounted games are also studied in context of dynamic systems~\cite{filar2012competitive}.
Positionally determined games in general  have a broad impact on 
 formal languages and automata theory~\cite{automata_toolbox}. 

\bigskip

\textbf{Value problem vs.~decision problem.} The value problem, as more general one, is at least as hard as the decision problem. On the other hand, the values in the discounted and mean payoff games can be obtained from a play of two positional strategies. Hence, the bit-length of values is polynomial in the bit-length of the edge weights   and (in case of discounted games) the discount factor. This makes the value problem polynomial-time reducible to the decision problem via binary search. For energy games there is no difference between these two problems at all.  

On the other hand, for discounted and mean payoff games the value problem may turn out to be harder  for \emph{strongly polynomial} algorithms. Indeed, a reduction from the value problem to the decision problem via the binary search does not work when the weights are arbitrarily real numbers.


\bigskip

\textbf{Reductions, structural complexity.} It is known that $\Max$ wins in an energy game if and only if the value of the corresponding mean payoff game is non-negative~\cite{bouyer2008infinite}. Hence, energy games are equivalent to the decision problem for mean-payoff games with threshold $0$. Any other threshold $\alpha$ is reducible to threshold  $0$ by adding $-\alpha$ to all the weights. So energy games and mean payoff games are polynomial-time equivalent.

Decision problem for discounted games lies in UP$\cap$coUP~\cite{jurdzinski1998deciding}.
In turn, mean payoff games are polynomial-time reducible to discounted games~\cite{zwick1996complexity}. Hence, the same  UP$\cap$coUP upper bound applies to mean payoff and energy games. None of these problems is known to lie in P.

\bigskip

\textbf{Algorithms for discounted games.} There are two classical approaches to discounted games. In \emph{value iteration} approach, going back to Shapley~\cite{shapley1953stochastic}, one manipulates with a real vector indexed by the nodes. The vector of values of a discounted game is known to be a fixed point of an explicit contracting operator. By applying this operator repeatedly to an arbitrary initial vector, one obtains a sequence converging to the vector of values. Using this, Littman~\cite{littman1996algorithms} gave a deterministic $O\left(\frac{n^{O(1)}\cdot L}{1 - \lambda} \log\left(\frac{1}{1 - \lambda}\right)\right)$-time algorithm solving the value problem for discounted games. Here $n$ is the number of nodes, $\lambda$ is the discount factor and $L$ is the bit-length of input. This gives a polynomial time algorithm for $\lambda = 1 - \Omega(1)$.

\emph{Strategy iteration} approach, going back to Howard~\cite{howard1960dynamic} (see also~\cite{rao1973algorithms}), can be seen as a sophisticated way of iterating positional strategies of players. Hansen et al.~\cite{hansen2013strategy} showed that strategy iteration solves the value problem for discounted games in deterministic  $O\left(\frac{n^{O(1)}}{1 - \lambda} \log_2\left(\frac{1}{1 - \lambda}\right)\right)$-time. Unlike Littman's algorithm,  for $\lambda = 1 - \Omega(1)$ this algorithm is \emph{strongly} polynomial.

More recently, \emph{interior point methods} we applied to discounted games \cite{hansen2013complexity}. As of now, however, these methods do not outperform the algorithm of Hansen et al.

In all these algorithms the running time depends on $\lambda$  (exponentially in the bit-length of $\lambda$). As far as we know, no deterministic algorithm with running time $2^{o(n\log n)}$ regardless of the value of $\lambda$ was known. One can get $2^{O(n\log n)}$-time by simply trying all possible positional strategies of one of the players. Our main result pushes this bound down to $2^{O(n)}$. More precisely, we show the following
\begin{theorem}
\label{thm:disc}
There is a deterministic algorithm, finding the values in a discounted game on a graph with $n$ nodes in $n^{O(1)}\cdot (2 + \sqrt{2})^n$-time and $n^{O(1)}$-space. The algorithm can be performed even when the discount factor and the edge weights are arbitrary real numbers (assuming basic arithmetic operations with them are carried out by an oracle)\footnote{If the running time were polynomial, we could simply say that this algorithm is ``strongly polynomial''. Unfortunately, it seems that there is no such well-established terminology for superpolynomial-time algorithms. }.
\end{theorem}

\begin{remark}One of the anonymous reviewers of this paper noticed that the complexity analysis in Theorem \ref{thm:disc} can be improved to $n^{O(1)} (1/\rho)^n$, where $\rho$ is a unique root of the equation $\rho^2 + 2\rho/(1 - \rho) = 1$ in the interval $(0, 1)$. Numerically, $1/\rho \approx 3.214$, while $2 + \sqrt{2} \approx 3.414$. We describe this suggestion after our proof of Theorem \ref{thm:disc}.
\end{remark}

We also obtain a better bound for a special case of discounted games, namely for \emph{bipartite} discounted games. We call a discounted game bipartite if in the underlying graph each edge is either an edge from a Max's node to a Min's node or an edge from a Min's node to a Max's node. In other words, in a bipartite discounted game players can  only make moves alternatively.
\begin{theorem}
\label{thm:bip_disc}
There is a deterministic algorithm, finding the values in a bipartite discounted game on a graph with $n$ nodes in $n^{O(1)}\cdot 2^n$-time and $n^{O(1)}$-space.
The algorithm can be performed even when the discount factor and the edge weights are arbitrary real numbers (assuming basic arithmetic operations with them are carried out by an oracle).
\end{theorem}
Our algorithm is the fastest  known \emph{deterministic} algorithm for discounted games when $\lambda \ge 1 - (2 + \sqrt{2} +\Omega(1))^{-n}$. For bipartite discounted games it is the fastest one for $\lambda \ge 1 - (2 + \Omega(1))^{-n}$. For smaller discounts, the algorithm of Hansen et al.~outperforms the bound we obtain for our algorithm. One should also mention that their algorithm is applicable to more general \emph{stochastic} discounted games, while our algorithm is not.

In addition, it is known that \emph{randomized} algorithms can solve discounted games faster, namely, in time $2^{O(\sqrt{n}\cdot \log n)}$ \cite{ludwig1995subexponential, halman2007simple, bjorklund2005combinatorial}. These algorithms are based on formulating discounted games as an \emph{LP-type} problem~\cite{matouvsek1996subexponential}.

\bigskip

\textbf{Algorithms for mean payoff and energy games.} In the literature on the mean payoff and energy games it is often assumed that the edge weights are integers. In this case, for a given game graph we denote by $W$ the largest absolute value of an edge weight.

Zwick and Paterson~\cite{zwick1996complexity} gave an algorithm solving  the value problem for mean payoff games in \emph{pseudopolynomial} time, namely, in time $O(n^{O(1)} \cdot W)$ (see also \cite{pisaruk1999mean}). Brim et al.~\cite{brim2011faster} improved the polynomial factor before $W$. 
 In turn, Fijalkow et al.~\cite{fijalkow2018complexity} slightly improved the dependence on $W$ (from $W$ to $W^{1 -1/n}$).

There are algorithms with running time depending on $W$  much better (at the cost that they are exponential in $n$).  Lifshits  and Pavlov~\cite{lifshits2007potential} gave a $O(n^{O(1)}\cdot 2^n)$-time algorithm for energy games (here the running time does not depend at all on $W$). Recently, Dorfman et al.~\cite{dorfman2019faster} pushed $2^n$ down to $2^{n/2}$ by giving a $O(n^{O(1)}\cdot 2^{n/2}\log W)$-time algorithm for energy games. They also announced that the $\log W$ factor can be removed. At the cost  of an extra $\log W$ factor these algorithms can be lifted to the value problem for mean payoff games.

All these algorithms are deterministic. As for randomized algorithms, the state-of-the-art is $2^{O(\sqrt{n}\log n)}$-time, the same as for discounted games.

We show that:
\begin{theorem}
\label{thm:energy}
There is a deterministic algorithm, finding in $n^{O(1)} 2^{n/2}$-time and $n^{O(1)}$-space all the nodes that are winning for $\Max$ in a given energy game with $n$ nodes. The algorithm can be performed even when the edge weights are arbitrary real numbers (assuming basic arithmetic operations with them are carried out by an oracle)
\end{theorem}
This certifies that indeed, as stated without a proof in~\cite{dorfman2019faster}, the $\log W$-factor before $2^{n/2}$ can be removed. Moreover, this certifies that for the $2^{n/2}$ bound we do not need an assumption that the edge weights are integral.

\begin{remark}
An updated online version of~\cite{dorfman2019faster} contains an algorithm that works for real weights as well.  In addition, this algorithm also computes the minimal \textbf{energy levels} for all nodes where Max wins (i.e., minimal $c$ such that Max can guaranty that $c + w_1 + w_2 + \ldots + w_n \ge 0$ for all $n$). Still, by giving a proof of Theorem \ref{thm:energy} along with our algorithm for discounted games, we hope to present these two results in the same framework.
\end{remark}

%

\subsection{Our technique}
Arguably, our approach arises more naturally for discounted games, yet it roots in the algorithm of Dorfman et al.~for energy games. 

For discounted games we iterate a real vector $x$ with coordinates indexed by the nodes of the graph, until $x$ coincides with the vector of values. Thus, our approach can also be called value iteration. However, it differs significantly from the classical value iteration, and we call it \emph{polyhedral} value iteration.

We rely on a well-known fact that the vector of values is a unique solution to the so-called \emph{optimality equations}. Optimality equations can be seen as a set of conditions  of the following form. First, they include a system of linear inequalities over $x$. Namely, this system contains for each edge an inequality between the values of its endpoints.
In addition, optimality equations state that every node has an out-going edge for which the corresponding inequality turns into an equality.

By forgetting about this additional condition we obtain a polyhedron containing the vector of values. We call this polyhedron \emph{optimality polyhedron}. Of course, besides the vector of values there are some other points too.

 We initialize $x$ by finding any point belonging to the optimality polyhedron. There is little chance that $x$ will satisfy optimality equations.  So until it does, we do the following. We compute a \emph{feasible} shift, i.e., a shift that does not immediately lead us outside the optimality polyhedron.
 Then we move from $x$ along this shift as far as possible, until the border of the optimality polyhedron is reached. This point on the border will be the new value of $x$. 

We choose a shift in a very specific way. First, a feasible shift from $x$ should not violate \emph{tight edges} (edges for which the corresponding inequality in the optimality polyhedron turns into an equality on $x$). To produce such a shift, we consider an auxiliary discrete game which we call \emph{discounted normal play game} (DNP game for short), played only on the tight edges. Essentially, the DNP game for $x$ will be just our initial discounted game, but with zero edge weights and only with edges that are tight on $x$. It will be easy to see that the vector of values of a DNP game always forms a feasible shift. One might find this approach resembling the Primal-Dual approach from Combinatorial Optimization (see, e.g.~\cite{combopt}).

It turns out that this process converges to the vector of values. Moreover, it does in $O(n (2 + \sqrt{2})^n)$ steps. This bound follows from a certain combinatorial fact about directed graphs. This fact can be stated and proved independently of our algorithm. In more detail, we define a notion of a \emph{DNP games iteration}. A DNP games iteration is a sequence of DNP games on directed graphs, where each next DNP game is obtained from the previous one according to certain rules.
Once again, we do not assume that DNP games in a DNP games iteration are necessarily played on graphs of tight edges from our algorithm -- they can be played on arbitrary graphs.

After we give a formal definition of a DNP games iteration, we split the complexity analysis into two independent parts. First, we show that the sequence of DNP games arising from our algorithm always forms a DNP games iteration.
Second, we show that the length of \emph{any} DNP games iteration is  $O(n (2 + \sqrt{2})^n)$.  In the second (and the most challenging) part of the argument we do not mention at all our algorithm and our initial discounted game. We find this feature of our argument very important. First, it significantly simplifies the exposition. Second, it indicates that this sort of an argument might be relatively easy adapted to other games. We demonstrate this with the example of the algorithm of Dorfman et al. for energy games, and we hope that it might lead to some new applications.

\bigskip

Dorfman et al.~build upon a \emph{potential lifting} algorithm of  Brim et al.~\cite{brim2011faster}. They notice that in the algorithm of Brim et al.~a lot of consecutive iterations may turn out to be lifting  the same set of nodes. Instead, Dorfman et al.~perform all these iterations at once, accelerating the algorithm of Brim et al.

In our exposition, instead of potential lifting, we perform a polyhedral value iteration, but now for energy games. In this language, a series of repetitive actions in the algorithm of Brim et al.~corresponds to just one feasible shift in the polyhedral value iteration.

The underlying polyhedron will be called the \emph{polyhedron of potentials}. Loosely speaking, it will be a limit of optimality polyhedra as $\lambda\to 1$. This resembles a well-known representation of mean payoff games as a limit of discounted games, see, e.g.,~\cite{puterman2014markov}. 

 As we mentioned, the complexity analysis is also carried out by DNP games iteration. In fact, almost no new argument is needed. The same bound as for discounted games follows at no cost at all. To obtain an improvement to $2^{n/2}$, we just notice that in case of energy games case the underlying DNP games iteration satisfies some additional restrictions. 

%

\section{Preliminaries}

\subsection{Discounted games}

To specify a discounted game $\mathcal{G}$ one has to specify: a finite directed graph $G = (V, E)$,  in which every node has at least one out-going edge;  a partition of the set of nodes $V$ into two disjoint subsets $V_\Max$ and $V_\Min$; a weight function $w\colon E \to \mathbb{R}$; a real number $\lambda\in (0, 1)$ called the \emph{discount factor} of $\mathcal{G}$.

The following equations in $x\in\mathbb{R}^V$ are called \emph{optimality equations} for $\mathcal{G}$: 
\begin{align}
\label{max}
x_a &= \max\limits_{e = (a, b) \in E} w(e) + \lambda x_b, \qquad a \in V_{\Max},\\
\label{min}
x_a &= \min\limits_{e = (a, b) \in E} w(e) + \lambda x_b, \qquad a \in V_{\Min},
\end{align}

For any discounted game $\mathcal{G}$ there exists a unique solution $x^*$ to (\ref{max}--\ref{min})~\cite{shapley1953stochastic}. Such $x^*$ is called \emph{the vector of values} of $\mathcal{G}$ (see the Introduction for the game-theoretic interpretation of $x^*$). In Theorems \ref{thm:disc} and \ref{thm:bip_disc} we study an algorithmic problem of finding a solution to (\ref{max}--\ref{min}) for a given discounted game $\mathcal{G}$. We refer to this problem as to ``finding the values of a discounted game''.

\subsection{Energy games}

To specify an energy game $\mathcal{G}$, we need exactly the same information as for a discounted game, except that in energy games there is no discount factor. So below $G = (V, E), V_\Max, V_\Min$ and $w$ are exactly as in the previous subsection.

Let us give a formalization of what does it mean that ``a node $v\in V$ is winning for Max in the energy game $\mathcal{G}$''. Again, we refer the reader to the Introduction for a less formal discussion of this notion.

We use the following terminology for cycles in $G$.  By the weight of a cycle we mean the sum of the values of $w$ over its edges. We call a cycle positive if its weight is positive. In the same way we define negative cycles, zero cycles, and so on.

A \emph{Max's positional strategy} is a mapping $\sigma\colon V_\Max \to E$ such that for every $a\in V_\Max$ the edge $\sigma(a)$ starts in the node $a$. Similarly, a \emph{Min's positional strategy} is a mapping $\tau\colon V_\Min \to E$ such that for every $b\in V_\Min$ the edge $\tau(b)$ starts in the node $b$.

For a Max's positional strategy $\sigma$, we say that an edge $e = (a, b)\in E$ is \emph{consistent} with $\sigma$ if either $a\in V_\Min$ or $a\in V_\Max, e = \sigma(a)$. We denote by $E^\sigma$ the set of edges that are consistent with a Max's positional strategy $\sigma$. We define consistency with a Min's positional strategy $\tau$ similarly. We denote by $E_\tau$ the set of all edges that are consistent with a Min's positional strategy $\tau$.

For a Max's positional strategy $\sigma$ and a Min's positional strategy $\tau$, by $G^\sigma, G_\tau$ and $G^\sigma_\tau$ we denote the following three graphs: $G^\sigma = (V, E^\sigma), G_\tau = (V, E_\tau), G^\sigma_\tau = (V, E^\sigma\cap E_\tau)$.  

We say that a node $v\in V$ is \emph{winning for Max in the energy game $\mathcal{G}$} if there exists a Max's positional strategy $\sigma$ such that only non-negative cycles are reachable from $v$ in the graph $G^\sigma$. We say that a node $u\in V$ is \emph{winning for Min in the energy game $\mathcal{G}$} is there exists a Min's positional strategy $\tau$ such that only negative cycles are reachable from $u$ in the graph $G_\tau$. It is not hard to see that no node can be simultaneously winning for Max and winning for Min. Indeed, otherwise a unique simple cycle, reachable from this node in the graph $G^\sigma_\tau$, would be simultaneously non-negative and negative.
In fact, every node is either winning for Max or is winning for Min (this easily follows from the positional determinacy of the mean payoff games, see~\cite{bouyer2008infinite,chakrabarti2003resource,ehrenfeucht1979positional}). In Theorem \ref{thm:energy}, we consider an algorithmic problem of find the set of nodes that are winning for Max in a given energy game $\mathcal{G}$ (the rest of the nodes should be then winning for Min).

\begin{remark}
We call a discounted or an energy game $\mathcal{G}$ \emph{bipartite} if  $E\subseteq V_\Max\times V_\Min \cup V_\Min\times V_\Max$. In this case we also use a term   ``bipartite'' for the underlying graph of the game.
\end{remark}

\section{$n^{O(1)}\cdot (2 + \sqrt{2})^n$-time algorithm for discounted games}

In this section we give an algorithm establishing Theorem \ref{thm:disc} and \ref{thm:bip_disc}.
We consider a discounted game $\mathcal{G}$, given by  a graph $G = (V, E)$, a partition $V = V_\Max\sqcup V_\Min$, a weight function $w\colon E\to\mathbb{R}$, and a discount factor $\lambda$. We assume that $G$ has $n$ nodes and $m$ edges.

In Subsection \ref{subsec:dnpg} we define auxiliary games that we call \emph{discounted normal play games}. We use these games both in the formulation of the algorithm and in the complexity analysis. In Subsection \ref{subsec:opt_pol} we define the so-called \emph{optimality polyhedron} by relaxing optimality equations (\ref{max}--\ref{min}).

The algorithm is given in Subsection \ref{subsec:alg}. 
In the algorithm we iterate the points of the optimality polyhedron in search of the solution to  (\ref{max}--\ref{min}). First, we initialize by finding any point belonging to the optimality polyhedron. Then for a current point we define a shift which does not immediately lead us outside the optimality polyhedron. In the definition of the shift we use discounted normal play games. To obtain the next point we move as for as possible along the shift until we reach the border.
 We do so until the current point satisfies (\ref{max}--\ref{min}). Along the way we also take some measures to prevent the bit-length of the current point from growing super-polynomially.

 This process always terminates and, in fact, can take only $O(n (2 + \sqrt{2})^n)$ iterations. Moreover, for bipartite discounted games it can take only $O(2^n)$ steps.
 A proof of it is deferred to Section \ref{sec:disc_complexity}.

\subsection{Discounted normal play games.}
\label{subsec:dnpg}
These games will always  be played on directed graphs with the same set of nodes as $G$. Given such a graph $G^\prime = (V, E^\prime)$, we equip it with the same partition of $V$ into  $V_{\Max}$ and $V_{\Min}$ as in $G$. There may be sinks in $G^\prime$, i.e., nodes with no out-going edges.

Two players called $\Max$ and $\Min$ move a pebble along the edges of $G^\prime$. Player $\Max$ controls the pebble in the nodes from $V_{\Max}$ and player $\Min$ controls the pebble in the nodes from $V_\Min$. If the pebble reaches a sink of $G^\prime$ after $s$ moves, then the player who cannot make a move pays a fine of size $\lambda^s$ to his opponent. Here $\lambda$ is the discount factor of a discounted game which we want to solve.
 If the pebble never reaches a sink, i.e., if the play lasts infinitely long, then players pay each other nothing. 

 By the outcome of the play we mean the income of player $\Max$. Thus, the outcome is 
\begin{itemize}
\item positive, if the play ends in a sink from $V_{\Min}$;
\item zero, if the play lasts infinitely long;
\item negative, if the play ends in a sink from $V_{\Max}$.
\end{itemize}
It is not hard to see that in this game players have optimal positional strategies. Moreover, if $\delta(v)$ is the value of this game in the node $v$, then
\begin{align}
\label{max_sink}
\delta(s) &= -1, \mbox{if $s$ is a sink from $V_\Max$},\\
\label{min_sink}
\delta(s) &= 1, \mbox{if $s$ is a sink from $V_\Min$},\\
\label{max_non_sink}
\delta(a) &= \lambda \cdot \max\limits_{(a, b) \in E^\prime} \delta(b), \mbox{if $a\in V_\Max$ and $a$ is not a sink},\\
\label{min_non_sink}
\delta(a) &= \lambda \cdot \min\limits_{(a, b) \in E^\prime} \delta(b), \mbox{if $a\in V_\Min$ and $a$ is not a sink}.
\end{align}
We omit proofs of these facts as below we only require the following
\begin{proposition}
\label{dnpg_values}
For any $G = (V, E^\prime)$ there exists exactly one solution to (\ref{max_sink}--\ref{min_non_sink}), which can be found in strongly polynomial time. 
\end{proposition}
Before proving Proposition \ref{dnpg_values} let us note that for graphs with $n$ nodes any solution $\delta$ to (\ref{min_sink}--\ref{max_non_sink}) satisfies $\delta(v) \in \{1, \lambda, \ldots, \lambda^{n -1}, 0, -\lambda^{n - 1}, \ldots, -1\}$. Indeed, if $a$ is not a sink, then by (\ref{max_non_sink}--\ref{min_non_sink}) the node $a$ has an out-going edge leading to a node with $\delta(b) = \delta(a)/\lambda$. By following these edges we either reach a sink after at most $n - 1$ steps (and then $\delta(a) = \pm \lambda^{i}$ for some $i\in\{0, 1, \ldots, n - 1\}$) or we go to a loop. For all the nodes on a loop of length $l \ge 1$ we have $\delta(b) = \lambda^l \delta(b)$, which means that $\delta(b) = 0$ everywhere on the loop (recall that $\lambda \in (0, 1)$). Thus, if we reach such a loop from $a$, we also have $\delta(a) = 0$. 

From this it is also clear that $\delta(v) = 1$ if and only if $v\in V_\Min$ and $v$ is a sink of $G^\prime$. Similarly, $\delta(v) = -1$ if and only if $v\in V_\Max$ and $v$ is a sink of $G^\prime$.

\begin{proof}[Proof of Proposition \ref{dnpg_values}]
To show the existence of a solution and its uniqueness we employ Banach fixed point theorem. Let $\Delta$ be the set of all vectors $f\in \mathbb{R}^V$, satisfying 
$$f(s) = 1 \mbox{ for all sinks $s\in V_\Min$}, \qquad f(t) =  -1 \mbox{ for all sinks $t\in V_{\Max}$}.$$
Define the following mapping $\rho\colon \Delta \to \Delta$:
$$
\rho(f)(a) = \begin{cases} -1 & \mbox{$a$ is a sink from $V_\Max$}, \\1 & \mbox{$a$ is a sink from $V_\Min$}, \\ \lambda \cdot \max\limits_{(a, b) \in E^\prime} f(b) & \mbox{$a\in V_\Max$ and $a$ is not a sink},\\
\lambda \cdot \min\limits_{(a, b) \in E^\prime} f(b) & \mbox{$a\in V_\Min$ and $a$ is not a sink}.\end{cases}
$$
The set of solutions to (\ref{max_sink}--\ref{min_non_sink}) coincides with the set of $\delta \in \Delta$ such that $\rho(\delta) = \delta$. It remains to notice that $\rho$ is $\lambda$-contracting with respect to $\|\cdot\|_\infty$-norm. 

Now let us explain how to find the solution to (\ref{max_sink}--\ref{min_non_sink}) in strongly polynomial time. In fact, the algorithm will be independent of the value of $\lambda$. Let us first determine for every $k\in\{0, 1, \ldots, n - 1\}$ the set $V_k = \{v \in V\mid \delta(v) = \lambda^k\}$. It is clear that $V_0$ coincides with the set of sinks of the graph $G^\prime$ that lie in $V_\Min$. Next, the set $V_k$ can be determined in strongly polynomial time provided $V_0, V_1, \ldots, V_{k - 1}$ are given. Indeed, by (\ref{max_non_sink}--\ref{min_non_sink}) the set $V_k$ consists of 
\begin{itemize}
\item all $v\in V_\Max\setminus V_{<k}$ that have an out-going edge leading to $V_{<k}$;
\item all $v\in V_\Min\setminus V_{<k}$ such that all edges starting at $v$ lead to $V_{<k}$.
\end{itemize}
Here $V_{<k} = V_0 \cup V_1 \cup\ldots \cup V_{k - 1}$. In this way we determine all the sets $V_0, V_1, \ldots, V_{n - 1}$. Similarly, one can determine all the nodes with $\delta(v) < 0$, and also the exact value of $\delta$ in these nodes.  All the remaining nodes satisfy $\delta(v) = 0$.
\end{proof}

\subsection{Optimality polyhedron}
\label{subsec:opt_pol}

By the
\emph{optimality polyhedron} of the discounted game $\mathcal{G}$ we mean the set of all $x\in \mathbb{R}^V$, satisfying the following inequalities:
\begin{align}
\label{max_pol}
x_a &\ge  w(e) + \lambda x_b  \qquad \mbox{ for } (a, b) \in E, a\in V_\Max,\\
 \label{min_pol}
x_a &\le w(e) + \lambda x_b  \qquad \mbox{ for } (a, b) \in E, a\in V_\Min.
\end{align}
We denote the optimality polyhedron by $\OptPol$.
Note that the solution to the optimality equations (\ref{max}--\ref{min}) belongs to $\OptPol$.


We call a vector $\delta\in \mathbb{R}^V$ a \emph{feasible shift} for $x\in\OptPol$ if for all small enough $\varepsilon > 0$ the vector $x+\varepsilon \delta$ belongs to $\OptPol$. To determine whether a shift $\delta$ is feasible for $x$ it is enough to look at the edges that are \emph{tight} for $x$. Namely, we call an edge  $(a, b) \in E$  tight for $x\in\OptPol$ if $ x_a =  w(e) + \lambda x_b$, i.e., if the corresponding inequality in (\ref{max_pol}--\ref{min_pol}) becomes an  equality on $x$. It is clear that $\delta\in \mathbb{R}^V$ is feasible for $x$ if and only if 
\begin{align}
\label{valid_max}
\delta(a) \ge \lambda\delta(b) \mbox{ whenever }  (a, b) \in E, a\in V_\Max \mbox{ and $(a, b)$ is tight for $x$},\\
\label{valid_min}
\delta(a) \le \lambda\delta(b) \mbox{ whenever }  (a, b) \in E, a\in V_\Min \mbox{ and $(a, b)$ is tight for $x$}.
\end{align}
Discounted normal play games can be used to produce for any $x\in\OptPol$  a feasible shift for $x$. Namely, let $E_x \subseteq E$ be the set of edges that are tight for $x$ and consider the graph $G_x = (V, E_x)$. I.e., $G_x$ is a subgraph of $G$ containing only edges that are tight for $x$. An important observation is that $x$ is the solution to optimality equations  (\ref{max}--\ref{min}) if and only if in $G_x$ there are no sinks.

Define $\delta_x$ to be the solution to (\ref{max_sink}--\ref{min_non_sink}) for $G_x$. It is easy to verify that the conditions (\ref{max_non_sink}--\ref{min_non_sink}) for $\delta_x$ imply (\ref{valid_max}--\ref{valid_min}), so $\delta_x$ is a feasible shift for $x$. Not also that as long as $x$ does not satisfy  (\ref{max}--\ref{min}), i.e., as long as the graph $G_x$ has sinks, the vector $\delta_x$ is not zero.

Let us also define a procedure $RealizeGraph(S)$ that we use in our algorithm to control the bit-length of the current point. In the definition of $RealizeGraph(S)$ we rely on the following result from~\cite{megiddo1983towards}. There exists a strongly polynomial-time algorithm $A$ that, given a system of linear inequalities with two variables per inequality, outputs a feasible solution to the system if the system is feasible, and outputs ``not found'' if the system is infeasible. We take any such $A$ and use it in the definition of $RealizeGraph(S)$. Namely, the input to the procedure $RealizeGraph(S)$ is a subset $S\subseteq E$. The output of $RealizeGraph(S)$ is the output of $A$ on a system that can be obtained from (\ref{max_pol}--\ref{min_pol}) by turning inequalities corresponding to edges from $S$ into equalities.

By definition, if the output of $RealizeGraph(S)$ is not ``not found'', then its output is a point $x\in\OptPol$ satisfying $S\subseteq E_x$. If the output of $RealizeGraph(S)$ is ``not found'', then there is no such point. Another important feature of $RealizeGraph(S)$ is that the bit-length of its output is always polynomially bounded. Indeed, its output coincides with an output of a strongly polynomial-time algorithm $A$ on a polynomially bounded input.

\subsection{The algorithm}
\label{subsec:alg}
\begin{algorithm}[H]
\label{disc_alg}
\SetAlgoLined
\KwResult{The solution to optimality equations (\ref{max}--\ref{min}) }
 initialization: $x = RealizeGraph(\emptyset)$\;
 \While{$x$ does not satisfy (\ref{max}--\ref{min}) }{
Compute $\delta_x$ using Proposition \ref{dnpg_values}\;
$\varepsilon_{max}\gets$ the largest $\varepsilon\in (0, +\infty)$ s.t $x + \varepsilon \delta_x\in\OptPol$\;
$x\gets RealizeGraph(E_{x + \varepsilon_{max}\delta_x})$\;
 }
output $x$\;
 \caption{$n^{O(1)} \cdot (2 + \sqrt{2})^n$-time algorithm for discounted games}
\end{algorithm}

Some remarks: 
\begin{itemize}
\item the value of $\varepsilon_{max}$ can be found as in the simplex-method. Indeed, $\varepsilon_{\max}$ is the smallest $\varepsilon\in(0,+\infty)$ for which there exists an inequality in (\ref{max_pol}--\ref{min_pol}) which is tight for $x + \varepsilon \delta_x$ but not for $x$. Thus, to find $\varepsilon_{max}$ it is enough to solve at most $m$ linear one-variable equations and compute the minimum over positive solutions to these equations.

\item in fact,  $\varepsilon_{\max} < +\infty$ throughout the algorithm, i.e, we cannot move along $\delta_x$ forever. To show this, it is enough  to indicate $\varepsilon > 0$ and an inequality in (\ref{max_pol}--\ref{min_pol}) which is tight for $x + \varepsilon \delta_x$ but not for $x$. First, since  $x$ does not  yet satisfy the optimality equations (\ref{max}--\ref{min}), there exists a sink $s$ of the graph $G_x$. Assume that $s\in V_{\Max}$, the argument in the case $s\in V_\Min$ is similar.
In the graph $G$ every node has an out-going edge, so there exists an edge $e = (s, b)\in E$. The edge $(s, b)$ is not tight for $x$ (otherwise $s$ is not a sink of $G_x$). Hence $x_s > w(e) + \lambda x_b$.  The same inequality for $x + \varepsilon \delta_x$ looks as follows:
$$x_s + \varepsilon\delta_x(s) \ge w(e) +  \lambda x_b + \varepsilon\lambda\delta_x(b)$$
Since the node $s$ is a sink of $G_x$ from $V_\Max$, we have $\delta_x(s) = -1 < \lambda \delta_x(b)$. Therefore, the left-hand side of the last inequality decreases in $\varepsilon$ faster than the right-hand side.  So for some positive $\varepsilon$ the left-hand and the right-hand side will become equal. This will be  $\varepsilon$for which the edge $(s, b)$ is tight for $x + \varepsilon \delta_x$.

\item The procedure $RealizeGraph$ can never output ``not found'' in the algorithm. Indeed, we always run it on a set of the form $S = E_y$ for some $y\in \OptPol$. Of course, for such $S$ there exists a point $x\in\OptPol$ such that $S\subseteq E_x$ -- for example, the point $y$ itself.

In addition, note that $x$ is always an output of the procedure $RealizeGraph$, so, as we discussed above, its bit-length is polynomially bounded throughout the algorithm.


\end{itemize}

\section{Discounted games: complexity analysis}
\label{sec:disc_complexity}

Let $x_0, x_1, x_2,\ldots$ be a sequence of points from $\OptPol$ that arise in the Algorithm \ref{disc_alg}. The argument consists of two parts:
\begin{itemize}
\item first, we show that the sequence of graph $G_{x_0}, G_{x_1}, G_{x_2}\ldots$ can be obtained in an abstract process that we call discounted normal play games iteration (DNP games iteration for short), see Subsection \ref{subsec:dnp_can};
\item second, we show that any sequence of $n$-node graphs that can be obtained in a DNP games iteration has length $O(n (2 + \sqrt{2})^n)$, see Subsection \ref{subsec:dnp_length}.
\end{itemize}
This will establish Theorem \ref{thm:disc}. In Subsection \ref{subsec:suggestion} we explain a suggestion of one of the reviewers of this paper, improving a bound on the length of a DNP games iteration to $n^{O(1)} \cdot 3.214^n$. 

 As for Theorem \ref{thm:bip_disc}, note that if $G$ is bipartite, then so are $G_{x_0}, G_{x_1}, G_{x_2}$, and so on. Thus, it is enough to demonstrate that:
\begin{itemize}
\item any sequence of \emph{bipartite} $n$-node graphs that can be obtained in a DNP games iteration has length $O(2^n)$, see Subsection \ref{subsec:dnp_bip}.
\end{itemize}

First of all, we have to give a definition of a DNP games iteration (Subsection \ref{subsec:dnp_def}).

\subsection{Definition of a DNP games iteration}
\label{subsec:dnp_def}
Consider a directed graph $H = (V = V_\Max\sqcup V_\Min, E_H)$ and let $\delta_H$ be the solution to (\ref{max_sink}--\ref{min_non_sink}) for $H$. We say that the edge $(a, b)\in E_H$ is \emph{optimal} for $H$ if $\delta_H(a) = \lambda\delta_H(b)$. Next, we say that a pair $(a, b)\in V\times V$ is \emph{violating} for $H$ if one of the following two conditions holds:
\begin{itemize}
\item $a\in V_{\Max}$ and $\delta_H(a) < \lambda \delta_H(b)$;
\item $a\in V_{\Min}$ and $\delta_H(a) > \lambda \delta_H(b)$.
\end{itemize}
Note that a violating pair of nodes cannot be an edge of $H$ because of (\ref{max_non_sink}--\ref{min_non_sink}).

Consider another directed graph $K = (V = V_\Max\sqcup V_\Min, E_K)$ over the same set of nodes as $H$, and with the same partition $V = V_\Max\sqcup V_\Min$. We say that  \emph{$K$ can be obtained from $H$ in one step of DNP games iteration} if the following two conditions hold:
\begin{itemize}
\item any optimal edge of $H$ is in $E_K$;
\item there is a pair of nodes in $E_K$ which is violating for $H$.
\end{itemize}
I.e., to obtain $K$ we can first erase some  (not necessarily all) non-optimal edges of $H$, and then we can add some edges that are not in $H$, in particular, we \emph{must} add at least one violating pair.

Finally, we say that a sequence of graph $H_0, H_1, \ldots, H_j$ \emph{can be obtained in a DNP games iterations} if for all $i\in\{0, 1, \ldots, j - 1\}$ the graph $H_{i + 1}$ can be obtained from $H_i$ in one step of DNP games iteration.

\subsection{Why the sequence $G_{x_0}, G_{x_1}, G_{x_2}, \ldots$ can be obtained in DNP games iteration}
\label{subsec:dnp_can}
Let $x$ and $x^\prime = RealizeGraph(E_{x+\varepsilon_{max}\delta_x})$ be two consecutive points of $\OptPol$ in the algorithm. We have to show that the graph $G_{x^\prime}$ can be obtained from $G_x$ in one step of DNP games iteration. By definition of the procedure $RealizeGraph$ the graph $G_{x^\prime}$ contains all edges of the graph $G_y$, where $y =  x + \varepsilon_{max}\delta_x$. Hence it is enough to show the following:
\begin{enumerate}[label=\textbf{(\alph*)}]
\item all the edges of the graph $G_x$ that are optimal for $G_x$ are also in the graph $G_y$;
\item there is an edge of the graph $G_y$ which is a violating pair for the graph $G_x$.
\end{enumerate}

\bigskip

\textbf{Proof of (a).} Take any edge $(a, b)$ of the graph $G_x$ which is optimal for $G_x$. The corresponding inequality in (\ref{max_pol}--\ref{min_pol}) turns into an equality on $x$. Now, consider the same inequality for the point $y = x + \varepsilon_{max} \delta_x$. Its left-hand side will be bigger by $\varepsilon_{max}\cdot \delta_x(a)$, and its right-hand side will be bigger by $\varepsilon_{max} \cdot \lambda \delta_x(b)$. Since $(a, b)$ is optimal for $G_x$, these two quantities are equal. So the left-hand and the right-hand side will still be equal on $y$. Hence $(a, b)$ belongs to $G_y$.


\bigskip

\textbf{Proof of (b).} In fact, any edge of the graph $G_y$ which is not in the graph $G_x$ is a violating pair for $G_x$. Indeed,  assume that $(a, b)\in E$ is an edge of $G_y$ but not of $G_x$. Consider an inequality in (\ref{max_pol}--\ref{min_pol}) corresponding to the edge $(a, b)$. By substituting $y$ there we obtain an equality, and by substituting $x$ there we obtain a strict inequality. Subtract one from another. This will give us  $\varepsilon_{max}\cdot \delta_x(a) < \varepsilon_{max} \cdot \lambda \delta_x(b)$ if $a\in V_\Max$ and $\varepsilon_{max}\cdot \delta_x(a) > \varepsilon_{max} \cdot \lambda \delta_x(b)$ if $a\in V_\Min$. This means that $(a, b)$ is a violating pair for $G_x$.

%
%
%

It only remains to note that there \emph{exists} an edge of $G_y$ which is not an edge of $G_x$. Indeed, otherwise all inequalities that are tight for $y = x + \varepsilon_{max}\delta_x$ were tight already for $x$. Then $\varepsilon_{max}$  could be increased, contradiction.

%

\subsection{$O(n (2 + \sqrt{2})^n)$ bound on the length of DNP games iteration}
\label{subsec:dnp_length}
The argument has the following structure. 
\begin{itemize}
\item \emph{Step 1.} For a directed graph $H = (V, E_H)$ we define two vectors $f^H, g^H\in\mathbb{N}^{2n - 1}$. 

\item \emph{Step 2.} We define the \emph{alternating lexicographic ordering} -- this will be a linear ordering on the set $\mathbb{N}^{2n - 1}$.

\item \emph{Step 3.} We show that in each step of a DNP games iteration \textbf{\emph{(a)}} neither $f^H$ nor $g^H$ decrease and \textbf{\emph{(b)}} either $f^H$ or $g^H$  increase (in the alternating lexicographic ordering).

\item \emph{Step 4.} We bound the number of values $f^H$ and $g^H$ can take (over all directed graph $H$ with $n$ nodes). By step 3 this  bound  (multiplied by 2) is also a bound on the length of a DNP games iteration.

\end{itemize}

\bigskip

\emph{Step 1.}
 The first coordinate of the vector $f^H$ equals the number of nodes with $\delta_H(a) = 1$ (all such nodes are from $V_\Min$). The other $2n - 2$ coordinates are divided into $n - 1$ consecutive pairs. In the $i$th pair we first have the number of nodes from $V_\Max$ with $\delta_H(a) = \lambda^i$, and then the number of nodes from $V_\Min$ with $\delta_H(a) = \lambda^i$. 

The vector $g^H$ is defined similarly, with the roles of $\Max$ and $\Min$ and $+$ and $-$ reversed. The first coordinate of $g^H$ equals the number of nodes with $\delta_H(a) = -1$ (all such nodes are from $V_\Max$). The other $2n - 2$ coordinates are divided into $n - 1$ consecutive pairs. In the $i$th pair we first have the number of nodes from $V_\Min$ with $\delta_H(a) = -\lambda^i$, and then the number of nodes from $V_\Max$ with $\delta_H(a) = -\lambda^i$. 

\bigskip

\emph{Step 2.}
The alternating lexicographic ordering  is a lexicographic order obtained from the standard ordering of integers in the even coordinates and from the reverse of the standard ordering of integers in the odd coordinates. I.e., we say that a vector $u\in\mathbb{N}^{2n - 1}$ is smaller than a vector $v\in \mathbb{N}^{2n - 1}$ in the alternating lexicographic order if there exists $i\in \{1, 2, \ldots, 2n - 1\}$ such that $u_j = v_j$ for all $1\le j < i$ and 
$$\begin{cases}u_{i} > v_i &\mbox{ if $i$ is odd}, \\ u_i < v_i & \mbox{ if $i$ is even}. \end{cases}$$

For example,
$$(3, 2, 3) < (2, 3, 2), \qquad (2, 3, 1) > (2, 2, 7),$$
in the alternating lexicographic order on $\mathbb{N}^3$.

\bigskip

\emph{Step 3.} This step relies on the following

\begin{lemma}
\label{main_lemma}
 Assume that a graph $H_2$ can be obtained from a graph $H_1$ in one step of a DNP games iteration. Then
\begin{enumerate}[label=\textbf{(\alph*)}]
\item if for some $i\in\{0,1, \ldots, n - 1\}$ it holds that $\{a \in V\mid \delta_{H_1}(a) = \lambda^i\} \neq \{a \in V\mid \delta_{H_2}(a) = \lambda^i\}$, then $f^{H_2}$ is greater than $f^{H_1}$ in the alternating lexicographic order.
\item if for some $i\in\{0,1, \ldots, n - 1\}$ it holds that $\{a \in V\mid \delta_{H_1}(a) = -\lambda^i\} \neq \{a \in V\mid \delta_{H_2}(a) = -\lambda^i\}$, then $g^{H_2}$ is greater than $g^{H_1}$ in the alternating lexicographic order.
\end{enumerate}
\end{lemma}
Assume Lemma \ref{main_lemma} is proved.
\begin{itemize}
\item \textbf{Why neither $f^H$ nor $g^H$ can decrease?} If $f^{H_2}$ does not exceed $f^{H_1}$ in the alternating lexicographic order, then $\{a \in V\mid \delta_{H_1}(a) = \lambda^i\} = \{a \in V\mid \delta_{H_2}(a) = \lambda^i\}$ for every $i\in\{0, 1, \ldots, n - 1\}$ by Lemma \ref{main_lemma}. On the other hand, $f^{H_1}$ and $f^{H_2}$ are determined by these sets, so $f^{H_1} = f^{H_2}$. Similar argument works for $g^{H_1}$ and $g^{H_2}$ as well.

\item \textbf{Why either $f^H$ or $g^H$ increase?} Assume that neither $f^{H_2}$ is greater than $f^{H_1}$ nor $g^{H_2}$ is greater than $g^{H_1}$ in the alternating lexicographic order. By Lemma \ref{main_lemma} we have for every $i\in\{0, 1, \ldots, n - 1\}$ that $\{a \in V\mid \delta_{H_1}(a) = \lambda^i\} = \{a \in V\mid \delta_{H_2}(a) = \lambda^i\}$ and $\{a \in V\mid \delta_{H_1}(a) = -\lambda^i\} = \{a \in V\mid \delta_{H_2}(a) = -\lambda^i\}$. This means that the functions $\delta_{H_1}$ and $\delta_{H_2}$ coincide. On the other hand, there is an edge of $H_2$ which is violating for $H_1$, contradiction. 
\end{itemize}
We now proceed to a proof of Lemma \ref{main_lemma}. Let us stress that in the proof we do not use the fact that $H_2$ contains a violating pair for $H_1$. We only use the fact that $H_2$ contains all optimal edges of $H_1$.
\begin{proof}[Proof of Lemma \ref{main_lemma}]
We only prove \textbf{\emph{(a)}}, the proof of \textbf{\emph{(b)}} is similar. Let $j$ be the smallest element of $\{0, 1, \ldots, n - 1\}$ for which $\{a \in V\mid \delta_{H_1}(a) = \lambda^j\} \neq \{a \in V\mid \delta_{H_2}(a) = \lambda^j\}$. First consider the case $j = 0$. We claim that in this case the first coordinate of $f^{H_2}$ is smaller than the first coordinate of $f^{H_1}$. Indeed, $f^{H_1}_1$ is the number of sinks from $V_\Min$ in the graph $H_1$. In turn, $f^{H_2}_1$ is the number of sinks from $V_\Min$ in the graph $H_2$. On the other hand, all sinks of $H_2$ are also sinks of $H_1$. Indeed, nodes that are not sinks of $H_1$ have in $H_1$ an out-going optimal edge. All these edges are also in $H_2$. Hence $f^{H_2}_1 \le f^{H_1}_1$.  The equality is not possible because otherwise $\{a \in V\mid \delta_{H_1}(a) = 1\} = \{a \in V\mid \delta_{H_2}(a) = 1\}$, contradiction with the fact that $j = 0$.

Now assume that $j > 0$. Then the sets $\{v\in V\mid \delta_{H_1}(v) = \lambda^j\}$ and $\{v\in V\mid \delta_{H_2}(v) = \lambda^j\}$ are distinct.  There are two cases:
\begin{itemize}
\item \emph{First case}: $\{v\in V_\Max\mid \delta_{H_1}(v) = \lambda^j\} \neq \{v\in V_\Max\mid \delta_{H_2}(v) = \lambda^j\}$.

\item \emph{Second case}: $\{v\in V_\Min\mid \delta_{H_1}(v) = \lambda^j\} \neq \{v\in V_\Min\mid \delta_{H_2}(v) = \lambda^j\}$.  
\end{itemize}
We will show that  in the first case we have $f^{H_1}_{2j} < f^{H_2}_{2j}$, and in the second case we have $f^{H_1}_{2j + 1} > f^{H_2}_{2j + 1}$. This would prove that  $f^{H_2}$ exceeds $f^{H_1}$ in alternating lexicographic order (recall that the first $1 + 2(j - 1)$ coordinates of $f^{H_1}$ and $f^{H_2}$ coincide by minimality of $j$, and note that they also coincide in the $(2j)$th coordinate whenever the first case does not hold).


\textbf{Proving $f^{H_1}_{2j} < f^{H_2}_{2j}$ in the first case.} By definition, $f^{H_1}_{2j}$ is the size of the set $\{v\in V_\Max\mid \delta_{H_1}(v) = \lambda^j\}$, and $f^{H_2}_{2j}$ is the size of the set $\{v\in V_\Max\mid \delta_{H_2}(v) = \lambda^j\}$. So it is enough to show that the first set is a subset of the second set (in fact, it would be a strict subset because we already know that these sets are distinct).

In other words, it is enough to show that for any $a\in V_\Max$ with $\delta_{H_1}(a) = \lambda^j$ we also have $\delta_{H_2}(a) = \lambda^j$. By (\ref{max_non_sink}--\ref{min_non_sink}) there is an edge $(a,b)$ of the graph $H_1$ with $\delta_{H_1}(b) = \lambda^{j - 1}$. We also have that $\delta_{H_2}(b) = \lambda^{j - 1}$, because $\{v\in V\mid \delta_{H_1}(v) = \lambda^{j-1}\} = \{v\in V\mid \delta_{H_2}(v) = \lambda^{j-1}\}$ (due to minimality of $j$). On the other hand, since $\delta_{H_1}(a) = \lambda \delta_{H_1}(b)$, the edge $(a,b)$ is optimal for $H_1$. Hence this edge is also in the graph $H_2$. So in the graph $H_2$ there is an edge from $a\in V_{\Max}$ to a node $b$ with $\delta_{H_2}(b) = \lambda^{j - 1}$. Hence by \eqref{max_non_sink} we have $\delta_{H_2}(a) \ge \lambda^j$. It remains to show why it is impossible that $\delta_{H_2}(a) > \lambda^{j}$. Indeed, by the minimal choice of $j$, the sets $\{v \in V \mid \delta_{H_1}(v) > \lambda^j\}$ and $\{v \in V \mid \delta_{H_2}(v) > \lambda^j\}$ are the same (and $a$ by definition is not in the first set).


\textbf{Proving $f^{H_1}_{2j + 1} > f^{H_2}_{2j + 1}$ in the second case.} By the same argument, it is enough to show that any $a\in V_{\Min}$ with $\delta_{H_2}(a) = \lambda^j$ also satisfies $\delta_{H_1}(a) = \lambda^j$. It is clear that $\delta_{H_1}(a) \le \lambda^{j}$, because otherwise for some $i < j$ we would have that the sets  $\{v\in V\mid \delta_{H_1}(v) = \lambda^i\}$ and  $\{v\in V\mid \delta_{H_2}(v) = \lambda^i\}$ are distinct ($a$ would belong to the first set and not to the second one). This would give us a contradiction with the minimality of $j$. Thus, it remains to show that $\delta_{H_1}(a) \ge \lambda^j$. Assume that this is not the case, i.e., $\delta_{H_1}(a) \le \lambda^{j + 1}$. Since $a\in V_\Min$, the node $a$ is not a sink of $H_1$, as the value in a Min's sink is $1 > \lambda^{j + 1}$. Hence by \eqref{min_non_sink} there exists an edge $(a, b)$ in the graph $H_1$ with $\delta_{H_1}(b) = \delta_{H_1}(a)/\lambda \le \lambda^j$. Then we also have that $\delta_{H_2}(b) \le \lambda^j$, because by minimality of $j$ we have $\{v\in V\mid \delta_{H_1}(v) \ge \lambda^{j - 1}\} = \{v\in V\mid \delta_{H_2}(v) \ge \lambda^{j - 1}\}$ and hence $\{v\in V\mid \delta_{H_1}(v) \le \lambda^{j}\} = \{v\in V\mid \delta_{H_2}(v) \le \lambda^{j}\}$. But the edge $(a, b)$ is optimal for $H_1$, so the edge $(a, b)$ is also in the graph $H_2$. This means that in the graph $H_2$ there is an edge from $a$ to a node $b$ with $\delta_{H_2}(b) \le \lambda^j$. Hence by \eqref{min_non_sink} we have $\delta_{H_2}(a) \le \lambda^{j + 1}$, contradiction.
\end{proof}

\bigskip

\emph{Step 4.} Notice  that $f^H$ and $g^H$ belong to the set of all vectors $v\in\mathbb{N}^{2n - 1}$ satisfying:
\begin{align}
\label{sum}
\|v\|_1 &\le n,\\
\label{a}
v_1 = 0 &\implies v_2 = v_3 = \ldots = v_{2n - 1} = 0, \\
\label{b}
v_{2i} = v_{2i + 1} = 0 &\implies v_{2i + 2} = v_{2i + 3} = \ldots v_{2n - 1} = 0 \mbox{ for every $i\in\{1, \ldots, n - 2\}$}.
\end{align}
To see \eqref{sum} note that in our case the $l_1$-norm is just a sum of coordinates. By construction, the sum of coordinates of $f^H$ is the number of nodes with $\delta_H(a) > 0$ and the sum of coordinates of $g^H$ is the number of nodes with $\delta_H(a) < 0$.
The fact that $f^H$ satisfies (\ref{a}--\ref{b}) can be seen from the following observation: if $\{a\in V\mid \delta_H(a) =\lambda^i\} = \emptyset$,  then we also have $\{a\in V\mid \delta_H(a) =\lambda^j\} = \emptyset$ for every $j \in\{i + 1, i + 2, \ldots, n - 1\}$. Indeed, by (\ref{max_non_sink}--\ref{min_non_sink}) a node with $\delta_H(a) = \lambda^j$ has an edge leading to a node with $\delta_H(b) = \lambda^{j - 1}$. By continuing in this way we would reach a node with $\delta_H(a) = \lambda^i$, contradiction.

Thus, the desired upper bound on the length of a DNP games iteration follows from the following technical lemma.

\begin{lemma}
\label{fact}
The number of vectors $v\in\mathbb{N}^{2n - 1}$ satisfying (\ref{sum}--\ref{b}) is $O(n (2 + \sqrt{2})^n)$.
\end{lemma}
\begin{proof}
Let $\mathcal{A}$ be the set of $v\in\mathbb{N}^{2n - 1}$ satisfying (\ref{sum}--\ref{b}). For $v\in\mathcal{A}$ let $t(v)$ be the largest $t\in\{1,2, \ldots, n - 1\}$ such that $v_{2t} + v_{2t + 1} > 0$. If there is no such $t$ at all (i.e., if $v_2 = v_3 = \ldots = v_{2n - 1} = 0$), then define $t(v) = 0$.

Let $\mathcal{A}_t = \{v \in\mathcal{A} \mid t(v) = t\}$.
We claim that $|\mathcal{A}_t| \le (2 + \sqrt{2})^n$ for any $t$.  As $t(v)$ takes only $O(n)$ values, the lemma follows. 

The size of $\mathcal{A}_0$ is $n$, so we may assume that $t > 0$. Take any $\rho\in (0, 1)$. Observe that:
\begin{align*}
\rho^n |\mathcal{A}_t| &\le \sum\limits_{v\in\mathcal{A}_t} \rho^{\|v\|_1} =\sum\limits_{v\in\mathcal{A}_t} \rho^{v_1} \cdot \rho^{v_2 + v_3} \cdot \ldots \rho^{v_{2t} + v_{2t + 1}} \\
&\le \left(\sum\limits_{v_1 = 1}^\infty \rho^{v_1}\right) \cdot \left(\sum\limits_{(v_2, v_3) \in\mathbb{N}^2 \setminus\{(0,0)\}} \rho^{v_2 + v_3}\right) \cdot \ldots \cdot \left(\sum\limits_{(v_{2t}, v_{2t + 1}) \in\mathbb{N}^2 \setminus\{(0,0)\}} \rho^{v_{2t} + v_{2t + 1}}\right)\\
&= \left(\sum\limits_{a = 1}^\infty \rho^a\right) \cdot \left(\sum\limits_{(b, c) \in\mathbb{N}^2 \setminus\{(0,0)\}} \rho^{b + c}\right)^t.
\end{align*}
Indeed, the first inequality here holds because $\|v\|_1 \le n$ by \eqref{sum} for $v\in\mathcal{A}$. The second inequality holds because for $v\in\mathcal{A}$ with $t(v) = t$ we have $v_1 > 0$ by \eqref{a} and $v_{2i} +v_{2i + 1} > 0$ for every $i\in\{1, 2, \ldots, t\}$ by \eqref{b}.

Next, notice that for $\rho = 1 - \frac{1}{\sqrt{2}}$ we have:
$$ \left(\sum\limits_{a = 1}^\infty \rho^a\right) \cdot \left(\sum\limits_{(b, c) \in\mathbb{N}^2 \setminus\{(0,0)\}} \rho^{b + c}\right)^t \le 1$$.
Indeed,
$$ \sum\limits_{a = 1}^\infty \rho^a = \frac{\rho}{1 - \rho} = \sqrt{2} - 1 < 1,$$
$$\sum\limits_{(b, c) \in\mathbb{N}^2 \setminus\{(0,0)\}} \rho^{b + c}  = \frac{1}{(1 - \rho)^2} - 1 = 1.$$
Thus, we get $\rho^n |\mathcal{A}_t| \le 1$. I.e., $|\mathcal{A}_t| \le (1/\rho)^n = (2 + \sqrt{2})^n$, as required. 
\end{proof}

\subsection{Improving the analysis to $O(3.214^n)$}
\label{subsec:suggestion}
As was noticed by one of the reviewers of this paper, our proof of Lemma \ref{main_lemma} actually proves a stronger statement, namely:
\begin{lemma}
\label{improved_main_lemma} Let $<_{Lex}$ denote the standard lexicographic order on the set $\mathbb{Z}^n$.
 Assume that a graph $H_2$ can be obtained from a graph $H_1$ in one step of DNP games iteration. Then
\begin{enumerate}[label=\textbf{(\alph*)}]
\item if for some $i\in\{0,1, \ldots, n - 1\}$ it holds that $\{a \in V\mid \delta_{H_1}(a) = \lambda^i\} \neq \{a \in V\mid \delta_{H_2}(a) = \lambda^i\}$, then:
$$(-f^{H_1}_1, f^{H_1}_2 - f^{H_1}_3, \ldots, f^{H_1}_{2n - 2} - f^{H_1}_{2n -1}) <_{Lex}(-f^{H_2}_1, f^{H_2}_2 - f^{H_2}_3, \ldots, f^{H_2}_{2n - 2} - f^{H_2}_{2n -1}).$$
\item if for some $i\in\{0,1, \ldots, n - 1\}$ it holds that $\{a \in V\mid \delta_{H_1}(a) = -\lambda^i\} \neq \{a \in V\mid \delta_{H_2}(a) = -\lambda^i\}$, then:
$$(-g^{H_1}_1, g^{H_1}_2 - g^{H_1}_3, \ldots, g^{H_1}_{2n - 2} - g^{H_1}_{2n -1}) <_{Lex}(-g^{H_2}_1, g^{H_2}_2 - g^{H_2}_3, \ldots, g^{H_2}_{2n - 2} - g^{H_2}_{2n -1}).$$
\end{enumerate}
\end{lemma}

Hence, up to a factor of $2$, the length of a DNP games iterations is bounded by the number of $u = (u_1, u_2, \ldots, u_n)\in\mathbb{Z}^n$ such that for some $(v_1, v_2, \ldots, v_{2n - 1})\in\mathbb{N}^{2n - 1}$ satisfying (\ref{sum}--\ref{b}) it holds that
$$u = (-v_1, v_2 - v_3, \ldots, v_{2n - 2} - v_{2n - 1}).$$
Now we are left with a purely combinatorial problem of bounding the number of such $u$. By using  the ``generating function'' method as in Lemma \ref{fact}, it is not hard to obtain a bound $n^{O(1)} (1/\rho)^n$, where $\rho$ is a unique root of the equation $\rho^2 + 2\rho/(1 - \rho) = 1$ in the interval $(0, 1)$.

\subsection{$O(2^n)$ bound on the length of DNP games iteration for bipartite graphs}
\label{subsec:dnp_bip}

The proof differs only in the last step, where for bipartite graphs we obtain a better bound. In more detail, if $H$ is bipartite, then $f^H$ and $g^H$ in addition to (\ref{sum}--\ref{b}) satisfy the following property:
\begin{equation}
\label{add}
u_{2i} = 0 \mbox{ for even $i\in\{1, \ldots, n - 1\}$}, \qquad u_{2i + 1} = 0 \mbox{ for odd $i\in\{1, \ldots, n - 1\}$}.
\end{equation}
Indeed, for $f^H$ the condition \eqref{add} looks as follows:
\begin{align*}
f^H_{2i} &= |\{v\in V_\Max\mid \delta_H(v) = \lambda^i\}| = 0 \mbox{ for even $i$,} \\
f^H_{2i + 1} &= |\{v\in V_\Min\mid \delta_H(v) = \lambda^i\}| = 0 \mbox{ for odd $i$.} 
\end{align*}
This holds because from a node $a$ with $\delta_H(a) = \lambda^i$ there a path of length $i$ to a node $s$ with $\delta_H(s) = 1$. If $\delta_H(s) = 1$, then $s\in V_\Min$. Since $H$ is bipartite, this means that $a\in V_\Min$ for even $i$ and $a\in V_\Max$ for odd $i$. The argument for $g^H$ is the same.

So it is enough to show that the number of $v\in\mathbb{N}^{2n - 1}$ satisfying (\ref{sum}--\ref{add}) is $O(2^n)$. Let $t(v)$ be defined in the same way as in the proof of Lemma \ref{fact}. I.e., $t(v)$ is the largest $t\in\{1, \ldots, n - 1\}$ for which $v_{2t} + v_{2t + 1} > 0$ (if there is no such $t$, we set $t(v) = 0$). Let us bound the number of $v\in\mathbb{N}^{2n - 1}$ satisfying (\ref{sum}--\ref{add}) and $\|v\|_1 = s, t(v) = t$.

For $t = 0$ the number of such $v$ is exactly $1$. Assume now that $t > 0$. Then 
\begin{align*}
v_1 &> 0, && \mbox{by \eqref{a}}\\
v_{2i} &> 0 \mbox{ and } v_{2i + 1} = 0, \mbox{ for odd $i\in\{1, \ldots, t\}$} && \mbox{by \eqref{b} and \eqref{add},}\\
v_{2i} &= 0 \mbox{ and } v_{2i + 1} > 0, \mbox{ for even $i\in\{1, \ldots, t\}$} && \mbox{by \eqref{b} and \eqref{add},}\\
v_j &= 0 \mbox{ for } j > 2t + 1, && \mbox{by definition of $t(v)$.}
\end{align*}
Hence the number of $v\in\mathbb{N}^{2n - 1}$ satisfying (\ref{sum}--\ref{add}) and $\|v\|_1 = s, t(v) = t$ is equal to the number of the solutions to the following system:
$$x_1 + x_2 + \ldots + x_{t + 1} = s, \qquad x_1, x_2, \ldots, x_{t + 1} \in \mathbb{N}\setminus\{0\}.$$
This number is $\binom{s -1}{t}$. By summing over all $s \le n$ and $t$ we get the required $O(2^n)$ bound.


\section{$n^{O(1)} \cdot 2^{n/2}$-time algorithm for energy games}

In this section we give an algorithm establishing Theorem \ref{thm:energy}. We consider an energy game $\mathcal{G}$ on a graph $G = (V, E)$ with a weight function $w\colon E\to\mathbb{R}$ and with a partition of $V$ between the players given by the sets $V_{\Max}$ and $V_{\Min}$. We assume that $G$ has $n$ nodes and $m$ edges. 

First, we notice that without loss of generality we may assume that $\mathcal{G}$ is bipartite. 

\begin{lemma}
\label{red_to_bip}
An energy game on $n$ nodes can be reduced in strongly polynomial time to a bipartite energy game on at most $n$ nodes.
\end{lemma}
This fact seems to be overlooked in the literature. Here is a brief sketch of it. Suppose that the pebble is in $a\in V_\Max$. After controlling the pebble for some time  $\Max$ might decide to enter a $\Min$'s node $b$. Of course, it makes sense to do it via a path of the largest weight (among all paths from $a$ to $b$ with intermediate nodes controlled by $\Max$). We can simply replace this path by a single edge from $a$ to $b$ of the same weight. Similar thing can be done with $\Min$, but now the weight should be minimized. By performing this for all pair of nodes controlled by different players we obtain an equivalent bipartite game. A full proof is given in Appendix \ref{app:b}.

To simplify an exposition we first present our algorithm for the case when the following assumption is satisfied.
\begin{assumption}
\label{as}
In the graph $G$ there are no zero cycles.
\end{assumption}
 Discussion of the general case is postponed to the end of this section.

Exposition of the algorithm follows the same scheme as for discounted games. First we define a polyhedron that we will work with. Now we call it the \emph{polyhedron of potentials}. In the algorithm we iterate the points of this polyhedron via feasible shifts. To produce feasible shifts we again use discounted normal play games. We should also modify a terminating condition. Given a point of the polyhedron of potentials satisfying our new terminating condition, one should be able to find all the nodes that are winning for  $\Max$ in our energy game. We also describe an analog of procedure $RealizeGraph$ (again used to control the bit-length of points that arise in the algorithm). All this is collected together in the Algorithm \ref{energy_alg}. Here are details.

\bigskip

The polyhedron of potentials is defined as follows:
\begin{align}
\label{max_pot}
 x_a &\ge w(e) + x_b \qquad \mbox{ for } (a, b) \in E, a\in V_\Max,\\
 \label{min_pot}
  x_a &\le w(e) + x_b \qquad \mbox{ for } (a, b) \in E, a\in V_\Min.
\end{align}
Here $x$ is an $n$-dimensional real vector with coordinates indexed by the nodes of the graph. This polyhedron is denoted by $\PolPoten$.

By setting
$$x_a = \begin{cases} W & a\in V_\Max,\\ 0 & a\in V_\Min,\end{cases}$$
for $W = \max_{e\in E} |w(e)|$ we obtain that $\PolPoten$ is not empty (here it is important that our energy game is bipartite).

We use notions similar to those we gave for the optimality polyhedron. Namely, we call an edge $e = (a, b)\in E$ tight for $x\in\PolPoten$ if $x_a = w(e) + x_b$. The set of all $e\in E$ that are tight for $x\in \PolPoten$ is denoted by $E_x$. By $G_x$ we mean the graph $(V, E_x)$. A very important consequence of the Assumption \ref{as} is that for every $x\in \PolPoten$ the graph $G_x$ is a directed \emph{acyclic} graph. Indeed, a cycle consisting of edges that are tight for $x$ would be a zero cycle, contradicting Assumption \ref{as}.

Next, we call a vector $\delta\in\mathbb{R}^n$ a feasible shift for $x\in\PolPoten$ if for all small enough $\varepsilon > 0$ it holds that $x + \varepsilon \delta\in\PolPoten$. Again, discounted normal play games on $G_x$ can be used to produce a feasible shift for $x$. Now  the discount factor in a discounted normal play game is irrelevant. We can pick an arbitrary one, say, $\lambda = 1/2$. As before, for $x\in\PolPoten$ we let $\delta_x$ be the solution to (\ref{max_sink}--\ref{min_non_sink}) for the graph $G_x$. Since the graph $G_x$ is acyclic, we have $\delta_x(a) \neq 0$ for every $a\in V$. Define $V_x^+ = \{a\in V\mid \delta_x(a) > 0\}$ and $V_x^- = \{a \in V\mid \delta_x(a) < 0\}$.

\begin{lemma}
\label{shift}
Assume that $x\in\PolPoten$ and let $\chi_x^+$ be the characteristic vector of the set $V_x^+$. Then $\chi_x^+$ is a feasible shift for $x$.
\end{lemma} 
\begin{proof}
Assume that $(a, b) \in E_x$. It is enough to show that $\chi_x^+(a) \ge \chi_x^+(b)$ if $a\in V_\Max$ and  $\chi_x^+(a) \le \chi_x^+(b)$ if $a\in V_\Min$.

 First, assume that $a\in V_\Max$ and $\chi_x^+(a) < \chi_x^+(b)$. Then $\chi_x^+(a) = 0$ and $\chi_x^+(b) = 1$, i.e., $\delta_x(b) > 0$ and $\delta_x(a) < 0$. But this contradicts \eqref{max_non_sink}.

Similarly, assume that  $a\in V_\Min$ and $\chi_x^+(a) > \chi_x^+(b)$. Then $\chi_x^+(a) = 1$ and $\chi_x^+(b) = 0$, i.e., $\delta_x(b) < 0$ and $\delta_x(a) >0$. This contradicts \eqref{min_non_sink}.
\end{proof}


A pair of nodes $(a, b) \in V\times V$ is called \emph{strongly violating} for $x$ if either $a\in V_\Min, \delta_x(a) > 0, \delta_x(b) < 0$ or $a\in V_\Max, \delta_x(a) < 0, \delta_x(b) > 0$. Clearly, if $(a, b)$ is strongly violating for $x$, then $(a, b)$ does not belong to $G_x$.

The following lemma specifies and justifies our new terminating condition.
\begin{lemma}
\label{correct}
Let $x\in\PolPoten$ and assume that no edge of the graph $G$ is strongly violating for $x$. Then for the energy game $\mathcal{G}$, the set $V_x^+$ is the set of nodes that are winning for Max, and the set $V_x^-$ is the set of nodes that are winning for Min. 
\end{lemma}
\begin{proof}
Consider a positional strategy $\sigma$ of $\Max$ defined as follows. For all $a\in V_x^+\cap V_\Max$ strategy $\sigma$ goes from $a$ by an edge $(a, b) \in E_x$ with $b\in V_x^+$. There is always such an edge because of \eqref{max_non_sink} and because there are no sinks from $V_\Max$ in $V_x^+$.  In the nodes from $V_x^-\cap V_\Max$ strategy $\sigma$ can be defined arbitrarily.

Let us also define the following positional strategy $\tau$ of $\Min$. For all $a\in V_x^-\cap V_\Min$ strategy $\tau$ goes from $a$ by an edge $(a, b) \in E_x$ with $b\in V_x^-$. Again,  such an edge exists by \eqref{min_non_sink} and since there are no sinks from $V_\Min$ in $V_x^-$.  In the nodes from $V_x^+\cap V_\Min$ strategy $\tau$ can be defined arbitrarily.

First, let us verify that for every $a\in V_x^+$ from $a$ one can reach only non-negative cycles in the graph $G^\sigma$. This would mean that the nodes from $V_x^+$ are winning for $\Max$ in $\mathcal{G}$.
First, it is impossible to reach $V_x^-$ in $G^\sigma$ from $a\in V_x^+$. Indeed, $\sigma$ does not leave $V_x^+$ by definition. In turn, an edge that starts in a Min's node from $V_x^+$ and goes to $V_x^-$ would be strongly violating for $x$. Hence it is enough to show that in the graph $G^\sigma$ every cycle consisting of nodes from $V_x^+$ is non-negative. Note that we can compute the weight of a cycle by summing up $w(e) + x_b - x_a$ over all edges $e = (a, b)$ belonging to this cycle (the terms $x_a$ cancel out). In turn, for edges of $G^\sigma$ lying inside $V_x^+$ all expressions $w(e) + x_b - x_a$ are non-negative.
Indeed, for every $e$ that starts in $V_\Min$ the expression  $w(e) + x_b - x_a$ is non-negative by \eqref{min_pot}. In turn strategy $\sigma$ uses edges of the graph $G_x$, i.e., edges that are tight for $x$. For these edges we have $w(e) + x_b - x_a = 0$.

Similarly one can show that  for every $a\in V_x^-$ from $a$ one can reach only non-positive cycles in the graph $G_\tau$. In fact, by Assumption \ref{as} there are no zero cycles, so all these cycles will be strictly negative.
\end{proof}

 We also define an analogue of the procedure $RealizeGraph(S)$ that we used in the discounted case. Its input is again a subset $S\subseteq E$.  Its output will be either a point $x\in\PolPoten$ or ``not found''. This procedure will again have the following features.
\begin{itemize}
\item if its output is not ``not found'', then its output is a point $x\in\PolPoten$ such that $S \subseteq E_x$. If its output is ``not found'', then there is no such point.
\item The bit-length of an output of $RealizeGraph(S)$ is polynomially bounded.
\end{itemize}
To achieve this, one can again run Megiddo's algorithm~\cite{megiddo1983towards} on a system obtained from (\ref{max_pot}--\ref{min_pot}) by turning inequalities corresponding to edges from $S$ into equalities. In fact, this system has rather specific form now. Namely, all inequalities in (\ref{max_pot}--\ref{min_pot}) are of the form $x \le y + c$, where $x, y$ are variables and $c$ is a constant. For such systems one can use, for example, a simpler algorithm of Pratt~\cite{pratt}. It is also well-known that such systems are essentially equivalent to the shortest path problem.

\bigskip

%
%
%

Now we are ready to give an algorithm establishing Theorem \ref{thm:energy}. Our goal is to find the sets
\begin{align*}
W_\Max &= \{a \in V\mid \mbox{$a$ is winning for $\Max$ in the energy game $\mathcal{G}$}\}, \\
 W_\Min &= \{a \in V\mid \mbox{$a$ is winning for $\Min$ in the energy game $\mathcal{G}$}\}.
\end{align*}
\begin{algorithm}[H]
\label{energy_alg}
\SetAlgoLined
\KwResult{The sets $W_\Max, W_\Min$.}
 initialization: $x = RealizeGraph(\emptyset)$\;
 \While{ there is an edge of $G$ which is strongly violating for $x$}{
$\varepsilon_{max}\gets$ the largest $\varepsilon\in (0, +\infty)$ s.t $x + \varepsilon \chi_x^+\in\PolPoten$\;
$x\gets RealizeGraph(E_{x + \varepsilon_{max}\chi_x^+})$\;
 }
output $W_\Max = V_x^+,\,\, W_\Min = V_x^-$\;
 \caption{$n^{O(1)} \cdot 2^{n/2}$-time algorithm for energy games}
\end{algorithm}
The correctness of the output of our algorithm follows from Lemma \ref{correct}.
To compute $V_x^+, V_x^-$ and $\chi_x^+$, and to check the terminating condition we find  $\delta_x$ in strongly polynomial time by Lemma \ref{dnpg_values}.
In turn, we compute $\varepsilon_{max}$ in the same way as in Algorithm \ref{disc_alg}. To demonstrate the correctness of the algorithm it only remains to show that $\varepsilon_{max} < +\infty$ throughout the algorithm. Indeed, when the terminating condition is not yet satisfied, there exists an edge $e = (a, b)$ of the graph $G$ which is strongly violating for $x$. This edge is not tight for $x$. Let us show that this edge is tight for $x + \varepsilon \chi_x^+$ for some positive $\varepsilon$. This would show that $\varepsilon_{max} \le \varepsilon$.

Let us only consider the case $a\in V_\Max$, the case $a\in V_\Min$ will be similar. Since $(a, b)$ is not tight for $x$, we have $x_a > w(e) + x_b$. Consider the same inequality for $x + \varepsilon \chi_x^+$. Its left-hand side will be bigger by $\varepsilon \chi_x^+(a)$, while its right-hand side will be bigger by $\varepsilon \chi_x^+(b)$. Since $(a, b)$ is strongly violating for $x$, we have $\chi_x^+(b) = 1, \chi_x^+(a) = 0$. So the left-hand side does not change with $\varepsilon$, while the right-hand side strictly increases with $\varepsilon$. Hence they the left-hand and the right-hand sides will become equal for some positive $\varepsilon$.


\subsection{What if Assumption \ref{as} does not hold?}
Assume that we add small $\rho > 0$ to the weights of all edges. Then all non-negative cycles in $G$ become strictly positive. On the other hand, if $\rho$ is small enough, then all negative cycles stay negative. Thus, for all small enough $\rho > 0$ we obtain in this way an energy game equivalent to the initial one and satisfying Assumption \ref{as}. The problem is how to find $\rho > 0$ small enough so that this argument work.

If edge weights are integers, then we can set $\rho = 1/(n + 1)$. However, it is not clear how to find
 a suitable $\rho > 0$ when the weights are arbitrary real numbers. In this case we use a standard ``symbolic perturbation'' argument that allows to avoid finding $\rho$ explicitly.

Namely, we add $\rho$ to all weights of edges not as a real number but as a \emph{formal variable}. I.e., we will consider the weights as formal linear combinations of the form $a + b\cdot \rho$, where $a, b\in\mathbb{R}$ are coefficients. First, we will perform additions over such combinations. More specifically, the sum of $a + b\cdot\rho$ and $c + d\cdot \rho$ will be $(a + c) + (b + d) \cdot \rho$. We will also perform comparisons of these linear combinations. We say that $a + b\cdot \rho < c + d\cdot \rho$ if $a < c$ or $a = c, b < d$. Note that the inequality $a + b\cdot \rho < c + d\cdot \rho$ holds for \emph{formal linear combinations} $a + b\cdot \rho$ and $c + d\cdot \rho$ if and only if  for all small enough $\epsilon\in\mathbb{R}, \epsilon > 0$ the same inequality holds for \emph{real numbers} when one substitutes $\epsilon$ instead of $\rho$.

Thus, more formally, we consider the weights as elements of the additive group $\mathbb{R}^2$ equipped with the lexicographic order. Now, given our initial ``real'' energy game, we consider another one where the weight of an edge $e\in E$ is a formal linear combination $w(e) + \rho$. After that Assumption \ref{as} is satisfied (again, if one understands the weight of a cycle as an element of the group $\mathbb{R}^2$).

We then run Algorithm \ref{energy_alg}, but now with the coordinates of the vector $x$ being elements of the group $\mathbb{R}^2$. We will not have to perform any multiplications and divisions with our formal linear combinations\footnote{Multiplications and divisions would not be a disaster for this argument as we could consider formal rational fractions over $\rho$. However, we find it instructive to note that we never go beyond the group $\mathbb{R}^2$.}. 
This is because it is possible to perform Algorithm \ref{energy_alg} using only additions and comparisons.
Indeed, in computing $\varepsilon_{max}$ we solve at most $m$ one-variable linear equations with the coefficient before the variable being $1$. We should also explain how to perform $RealizeGraph(S)$ procedure using only additions and comparisons. For that we implement $RealizeGraph(S)$ using Pratt's algorithm~\cite{pratt}. This algorithm takes a systems of linear inequalities of the from $x\le y + c$ (where $x$ and $y$ are variables and $c$ is a constant). It either outputs a feasible point of the system or recognizes its infeasibility. It does so by performing the standard  Fourier–Motzkin elimnation and removing redundant inequalities. Namely, for two variables $x$ and $y$ among all inequalities of the form $x \le y + c$ it only remembers one with the smallest $c$. This keeps the number of inequalities $O(n^2)$. Thus, complexity of Pratt's algorithm is strongly polynomial, and clearly we only need additions and comparisons to perform it.

%

To argue that a version of  Algorithm \ref{energy_alg} with formal linear combinations is correct we use a sort of compactness argument. Fix some $N$ and ``freeze'' the algorithm after $N$ steps. Up to now only  finitely many comparisons of linear combinations over $\rho$ were performed. For all small enough real $\epsilon > 0$ all these comparisons will have the same result if one substitutes $\epsilon$ instead of $\rho$. So after $N$ steps the ``formal'' version of Algorithm \ref{energy_alg} will be in the same state as the ``real'' one, i.e., one where in advance we add a small enough real number $\epsilon$ to all the weights. In turn, for all small enough $\epsilon$ the ``real'' version terminates in $N = n^{O(1)} 2^{n/2}$ steps (see the next section) with the correct output to our initial energy game. It is important to note that a bound $N$ on the number of steps of the ``real'' algorithm is independent of $\epsilon$. Hence the ``formal'' version also terminates in at most $N = n^{O(1)} 2^{n/2}$ steps with the correct output.

\section{Energy games: complexity analysis}
The complexity analysis of Algorithm \ref{energy_alg} follows the same scheme as for discounted games. First, we define \emph{strong} DNP games iteration (a more restrictive version of DNP games iteration, see Subsection \ref{subsec:sdnp}). Then we consider a sequence $x_0, x_1, x_2, \ldots$ of points from $\PolPoten$ that arise in Algorithm \ref{energy_alg}. We show that the corresponding sequence of graphs $G_{x_0}, G_{x_1}, G_{x_2}, \ldots$ can be obtained in a strong DNP games iteration (Subsection \ref{subsec:energy_can}). Finally, we show that the length of a strong DNP games iteration is bounded by $O(2^{n/2})$ (Subsection \ref{subsec:sdnp_bound}).

\subsection{Definition of strong DNP games iteration}
\label{subsec:sdnp}
In a strong DNP games iteration all graphs are assumed to be bipartite and acyclic.

Consider a directed bipartite acyclic graph $H = (V = V_\Max\sqcup V_\Min, E_H)$. We say that a pair of nodes $(a, b)\in V\times V$ is \emph{strongly} violating\footnote{This notion already appeared in Lemma \ref{correct}, but only for graphs of the form $G_x, x\in\PolPoten$.} for $H$ if either $a\in V_\Min$, $\delta_H(a) > 0, \delta_H(b) < 0$ or $a\in V_\Max, \delta_H(a) < 0, \delta_H(b) > 0$. Here, as before, $\delta_H$ is the solution to (\ref{max_sink}--\ref{min_non_sink}) for $H$ (and for $\lambda = 1/2$). Note once again that for acyclic graphs we have $\delta_H(a) \neq 0$ for all $a\in V$.

Consider another directed bipartite acyclic graph $K = (V, E_K)$ over the same set of nodes as $H$, and with the same partition of $V$ into Max's nodes and into Min's nodes. We say that  \emph{$K$ can be obtained from $H$ in one step of strong DNP games iteration} if the following two conditions holds:
\begin{itemize}
\item any optimal edge of $H$ is in $E_K$;
\item the set $E_K$ contains a pair of nodes which is strongly violating for $H$.
\end{itemize}

Finally, we say that a sequence of directed bipartite acyclic graphs $H_0, H_1, \ldots, H_j$ \emph{can be obtained in strong DNP games iterations} if for all $i\in\{0, 1, \ldots, j - 1\}$ the graph $H_{i + 1}$ can be obtained from $H_i$ in one step of strong DNP games iteration.

\subsection{Why the sequence  $G_{x_0}, G_{x_1}, G_{x_2}, \ldots$ can be obtained in strong DNP games iteration}
\label{subsec:energy_can}
Consider any two consecutive points $x$ and $x^\prime = RealizeGraph(E_{x+\varepsilon_{max}\chi_x^+})$ of $\PolPoten$ from Algorithm \ref{energy_alg}. We shall show that the graph $G_{x^\prime}$ can be obtained from $G_x$ in one step of strong DNP games iteration. First, note that both of these graphs are bipartite (because the underlying energy game is bipartite) and acyclic (because of Assumption \ref{as}). Set $y =x+\varepsilon_{max}\chi_x^+$.  As the graph $G_{x^\prime}$ contains all edges of the graph $G_y$, it is enough to show the following

\begin{enumerate}[label=\textbf{(\alph*)}]
\item all the edges of the graph $G_x$ that are optimal for $G_x$ are also in the graph $G_y$;
\item there is an edge of the graph $G_y$ which is a strongly violating pair for the graph $G_x$.
\end{enumerate}

\textbf{Proof of (a).} Take any edge $(a, b)$ of the graph $G_x$ which is optimal for $G_x$. Clearly, the values of $\delta_x(a)$ and $\delta_x(b)$ are either both positive or both negative. Hence the shift $\chi_x^+$ increases both $x_a$ and $x_b$ by the same amount. This means that $(a, b)$ is still tight for $y$, i.e., $(a, b)$ is an edge of $G_y$.

\textbf{Proof of (b).} First, there exists an edge $e = (a, b)\in E$ which belongs to the graph $G_y$ and not to $G_x$. Indeed, otherwise all edges that are tight for $y$ were already tight for $x$, and hence $\varepsilon_{max}$ could be increased. It is enough to show now that any edge $(a, b)\in E_y\setminus E_x$ is strongly violating for $G_x$. Since $(a, b)$ is not tight for $x$, we have:
\begin{itemize}
\item $x_a > w(e) + x_b$ if $a\in V_\Max$;
\item $x_a < w(e) + x_b$ if $a\in V_\Min$.
\end{itemize}
On the other hand, since $(a, b)$ is tight for $y$, we have:
$$ x_a + \varepsilon_{max}\chi_x^+(a) = w(e) + x_b + \varepsilon_{max}\chi_x^+(b).$$
Hence $\chi_x^+(a) < \chi_x^+(b)$ if $a\in V_\Max$ and $\chi_x^+(a) > \chi_x^+(b)$ if $a\in V_\Min$. Recall that the vector $\chi_x^+$ is the indicator of the set of nodes where the value of $\delta_x$ is positive. This means that $(a, b)$ is strongly violating for $G_x$.

\subsection{$O(2^{n/2})$ bound on length of strong DNP games iteration}
\label{subsec:sdnp_bound}
Note that strong DNP games iteration is a special case of DNP games iteration. Hence all the results we established for DNP games iteration can be applied here. Since we are dealing with bipartite graphs, we already have the bound $O(2^n)$ proved in Subsection \ref{subsec:dnp_bip}. 

Let us first give an idea what causes an improvement from $2^n$ to $2^{n/2}$. Unlike discounted games, we now have a guarantee that every time a \emph{strongly} violating pair appears. This yields that in each step of a strong DNP games iteration \emph{both} vectors $f^H$ and $g^H$ increase in the alternating lexicographic order, not only one. Now, what provides that each time we have a strongly violating pair? Loosely speaking, Algorithm \ref{energy_alg} looks like Algorithm \ref{disc_alg} with $\lambda = 1$. For $\lambda = 1$ all the nodes from $\{v\in V\mid \delta_x(v) > 0\}$ are shifted by the same quantity, similarly for the set $\{v\in V\mid \delta_x(v) < 0\}$. Hence no violating pair can appear inside one of these sets. Instead, a new violating pair will be between these two sets, i.e., it will be strongly violating.

We now proceed to a formal argument. First, let us explain why the fact that both $f^H$ and $g^H$ increase each time leads to a $O(2^{n/2})$ bound. Note that $\|f^H\|_1= |\{a\in V\mid \delta_H(a) > 0\}|$ and $\|g^H\|_1 = |\{a\in V\mid \delta_H(a) < 0\}|$. Hence $\|f^H\|_1 + \|g^H\|_1 = n$. Therefore, if a strong DNP games iteration has length $l$, then either $\|f^H\|_1 \le n/2$ at least $l/2$ times or $\|g^H\|_1 \le n/2$ at least $l/2$ times. Hence there are at least $l/2$ different vectors $v\in\mathbb{N}^{2n - 1}$ satisfying (\ref{sum}--\ref{add}) and $\|v\|_1 \le n/2$. On the other hand, the number of such vectors is $O(2^{n/2})$. Indeed, as shown in Subsection \ref{subsec:dnp_bip} the number of $v\in\mathbb{N}^{2n - 1}$ satisfying (\ref{sum}--\ref{add}) and $\|v\|_1 = s, t(v) = t$ is $\binom{s - 1}{t}$. By summing over all $s\le n/2$ and $t$ we get the required $O(2^{n/2})$ bound.

It only remains to explain why both $f^H$ and $g^H$ increase in each step of a strong DNP games iteration. Let $H_1$ and $H_2$ be two consecutive graphs in a strong DNP games iteration. Assume first that $f^{H_2}$ is not greater than $f^{H_1}$ in the alternating lexicographic order. By Lemma \ref{main_lemma} for every $i\in\{0, 1, \ldots, n - 1\}$ it holds that $\{a\in V\mid \delta_{H_1}(a) = \lambda^i\} = \{a\in V\mid \delta_{H_2}(a) = \lambda^i\}$. In particular, $\{a\in V \mid \delta_{H_1}(a) > 0\} = \{a\in V \mid \delta_{H_2}(a) > 0\}$. Since $\delta_{H_1}$ and $\delta_{H_2}$ are non-zero in all nodes (again, this is because these graphs are acyclic), we also have $\{a\in V \mid \delta_{H_1}(a) < 0\} = \{a\in V \mid \delta_{H_2}(a) < 0\}$. Hence a pair $(a, b) \in V\times V$ is strongly violating for $H_1$ if and only if it is strongly violating for $H_2$. On the other hand, the graph $H_2$ contains as an edge a strongly violating pair for $H_1$, contradiction.

Exactly the same argument shows that $g^{H_2}$ is greater than $g^{H_1}$ in the alternating lexicographic order.


\section{Discussion}
We do not know whether the bounds we obtain for Algorithm \ref{disc_alg} are tight. It seems unlikely that this algorithm is  actually  subexponential. This is because an updated version of~\cite{dorfman2019faster} now contains a tight example for their algorithm. Due to similarities between these two algorithms, it seems plausible that this example can lifted to discounted games.

One can consider a generalization of the discounted games, namely, the \emph{multi-discounted} games (where, roughly speaking, each edge can have its own discount).
Let us note that the multi-discounted games  can also be solved by an analogue of Algorithm \ref{disc_alg}. However, we do not know whether this analogue has better time complexity than $2^{O(n\log n)}$. The reason why our analysis cannot be carried out for the multi-discounted games is that now there are super-linearly many possible values    in the underlying DNP games.  Improving $2^{O(n\log n)}$ time for the multi-discounted games is interesting on its own, but it would also have consequences for the \emph{weighted} mean payoff games (see~\cite{gimbert}).

Finally, let us mention that from Algorithm \ref{disc_alg} one can actually obtain an algorithm for the \emph{value} problem for mean payoff games. For that one should run this algorithm with $\lambda$ being a formal variable, behaving as if it were ``arbitrarily close'' to $1$. This will give us a solution to the optimality equations in a form of rational fractions in $\lambda$. It is classical that from these fractions one can extract the values of the corresponding mean payoff game. This implies that the \emph{value} problem for mean payoff games can be solved in $2^{O(n)}$ time, even when the weights are \emph{real} numbers (again, assuming an oracle access to them as in the strongly polynomial algorithms). It seems that two known exponential time algorithms for mean payoff games~\cite{lifshits2007potential, dorfman2019faster} do not have this feature. Namely, they reduce the value problem to the decision problem by a binary search. When the weights are arbitrary real numbers, this reductions does not work.

\bigskip

\textbf{Acknowledgments.} I am grateful to Pierre Ohlmann for giving a talk about~\cite{dorfman2019faster} at the University of Warwick, and to Marcin Jurdzinski for discussions. I am also grateful to the anonymous reviewers of SODA 2021 for helpful comments, in particular for a suggestion improving Theorem \ref{thm:disc}. Finally, I would like to thank the authors of~\cite{dorfman2019faster} for pointing out to an updated version of their paper.

\appendix


\section{Proof of Lemma \ref{red_to_bip}}
\label{app:b}
Let us call a node $a\in V$ of the graph $G$ \emph{trivial} in the following two cases:
\begin{itemize}
\item $a\in V_\Max$ and only nodes of $V_\Max$ are reachable from $a$;
\item $a\in V_\Min$ and only nodes of $V_\Min$ are reachable from $a$.
\end{itemize}
Next, let us call a cycle $C$ of the graph $G$ \emph{trivial} in the following two cases:
\begin{itemize}
\item cycle $C$ is non-negative and all its nodes are from $V_\Max$;
\item cycle $C$ is negative and all its nodes are from $V_\Min$.
\end{itemize}
First step of our reduction is to get rid of trivial nodes and cycles. Note that once we have detected a trivial node or a trivial cycle, we can determine the winner of our energy game in at least one node. Indeed, to determine the winner in a trivial node we essentially need to solve a \emph{one-player} energy game. It is well-known that this can be done in strongly polynomial time. In turn, all nodes of a trivial cycle are winning for the player controlling these nodes -- he can win just by staying on the cycle forever.

Next, once the winner is determined in at least one node, there is a standard way of reducing the initial game to a game with fewer nodes. Suppose we know the winner in a node $a$, say, it is $\Max$. Then $\Max$ also wins in all the nodes from where he can enforce reaching $a$. We simply remove all these nodes. This does not affect who wins the energy games in the remaining nodes. Indeed, $\Max$ has no edges to removed nodes, and a winning strategy of $\Min$ would never use an edge to these nodes. It should be also noted that in the remaining graph all the nodes still have at least one out-going edge (a sink would have been removed).

So getting rid of trivial nodes and cycles can be done as follows. We first detect whether they exist. Then we determine the winner in some node of the graph and reduce our game to a game with smaller number of nodes. Clearly, all these actions take strongly polynomial time. This can be repeated at most $n$ times, so the whole procedure takes strongly polynomial time.

From now we assume that we are given an energy game $\mathcal{G}$ on a graph $G = (V, E)$ with no trivial cycles and nodes. We construct a bipartite graph $G^\prime$ over the same set of nodes and the corresponding bipartite energy game $\mathcal{G}^\prime$ equivalent to the initial one. In the definition of $G^\prime$ we use the following notation. Consider a path $p$ of the graph $G$. We say that $p$ is $\Max$-controllable if all the nodes of $p$ except the last one are from $V_\Max$ (the last one can belong to $V_\Min$ as well as to $V_\Max$). In other words, $\Max$ should be able to navigate the pebble along $p$ without giving the control to $\Min$. Similarly, we say that $p$ is $\Min$-controllable if all the nodes of $p$ except the last one are from $V_\Min$.

 First, consider a pair of nodes $a\in V_\Max, b\in V_\Min$. We include $(a, b)$ as en edge to the graph $G^\prime$ if and only if in $G$ there is a $\Max$-controllable path from $a$ to $b$. Since $a$ is not a trivial node in $G$, there will be at least one edge starting at $a$ in $G^\prime$. Provided $(a, b)$ was included, we let its weight in $G^\prime$ be the largest weight of a $\Max$-controllable path from $a$ to $b$ in $G$ (with respect to the weight function of $G$). We call a path on which this maximum is attained \emph{underlying} for the edge $(a, b)$.  The weight of $(a, b)$ in $G^\prime$ will be finite since in $G$ there are no positive cycles consisting entirely of nodes from $V_\Max$.

We have described edges of $G^\prime$ from $V_\Max$ to $V_\Min$. Edges in the opposite direction are defined analogously. Namely, consider a pair of nodes $a\in V_\Min, b\in V_\Max$. We include this pair to $G^\prime$ as an edge if and only if in $G$ there is a $\Min$-controllable path from $a$ to $b$. Once $(a, b)$ is included, we let its weight be the minimal weight of a $\Min$-controllable path from $a$ to $b$ in $G$. A path attaining this minimum will be called underlying for $(a, b)$. Again, absence of trivial nodes guaranties that in $G^\prime$ the node $a$ will have at least one out-going edge. The weight of $(a, b)$ will be well-defined due to absence of trivial cycles.

It only remains to argue that $\mathcal{G}^\prime$ is equivalent to $\mathcal{G}$. Let $W_\Max$ ($W_\Min$) be the set of nodes that are winning for $\Max$ (for $\Min$) in $\mathcal{G}$. It is enough to show that the set $W_\Max$ (the set $W_\Min$) is winning for $\Max$ ($\Min$) in $\mathcal{G}^\prime$. We prove it only for $W_\Max$, the argument for $W_\Min$ is similar.

 Let $\sigma$ be a $\Max$'s optimal positional strategy in $\mathcal{G}$. Consider the following $\Max$'s positional strategy $\sigma^\prime$ in the graph  $G^\prime$ (this strategy will be winning for $\Max$ in the game $\mathcal{G}^\prime$ for the nodes from $W_\Max$). We will define it only for nodes in $W_\Max$. Given a Max's node $a\in W_\Max$, apply $\sigma$  to $a$ repeatedly until a node from $V_\Min$ is reached.  In fact, there is a possibility that from $a$ strategy $\sigma$ loops before reaching any Min's node. But then the corresponding cycle would be negative  (there are no trivial cycles). This would mean that  $\sigma$ is not winning for $\Max$ in $a$. So we conclude that indeed  by applying repeatedly $\sigma$ to $a$ we reach a node from $V_\Min$. Let this node from $V_\Min$ be $b$. Note that $(a, b)$ is an edge of $G^\prime$, because we have reached $b$ by a $\Max$-controllable path from $a$. We let $\sigma^\prime(a) = (a, b)$.

We shall prove that only non-negative cycles are reachable from $W_\Max$ in $(G^\prime)^{\sigma^\prime}$.
First, note that edges that $\sigma^\prime$ uses do not leave $W_\Max$. This is because by applying a winning Max's strategy repeatedly we cannot leave $W_\Max$ in $G$. Moreover, no Min's edge in $G^\prime$ can leave $W_\Max$. Indeed, otherwise Min could leave $W_\Max$ in $G$. Thus, it remains to argue that any cycle $C^\prime$ in $(G^\prime)^{\sigma^\prime}$, located in $W_\Max$, is non-negative. We do so by indicating in the graph $G^\sigma$ a cycle $C$, located in $W_\Max$ and having at most the same weight as $C^\prime$. As $C$ is non-negative, the same holds for $C^\prime$.

To obtain $C$ we replace each edge $(a, b)$ of $C^\prime$ by a certain path $p_{(a, b)}$  from $a$ to $b$ in $G^\sigma$. The path $p_{(a, b)}$ will never leave $W_\Max$ and its weight in $G$ will be at most the weight of the edge $(a, b)$ in $G^\prime$. 

If $(a,b)\in V_\Min\times V_\Max$, we let $p_{(a,b)}$ be the underlying path for $(a, b)$. Its weight in $G$ just equals the weight of $(a, b)$ in $G^\prime$. As this path is $\Min$-controllable, it belongs to $G^\sigma$ and thus never leaves $W_\Max$.

If   $(a,b)\in V_\Max\times V_\Min$, then $(a, b)$ is used by strategy $\sigma^\prime$ in the node $a$. Hence by definition of $\sigma^\prime$ there is a $\Max$-controllable path in $G^\sigma$ from $a$ to $b$. We let $p_{(a,b)}$ be this path. It never leaves $W_\Max$ as $\sigma$ cannot leave $W_\Max$. The weight of $(a, b)$ in $G^\prime$ is the largest weight of a $\Max$-controllable path from $a$ to $b$ in $G$, so the weight of $p_{(a,b)}$ can only be smaller.

\end{document}